\documentclass[11pt]{amsart}
\RequirePackage{amssymb}

\renewcommand{\geq}{\geqslant}
\renewcommand{\leq}{\leqslant}

\usepackage{amsaddr}
\usepackage{fullpage}
\usepackage{amsmath}
\usepackage{amsthm}
\usepackage{amssymb}
\usepackage{bm}
\usepackage{tikz}
\usepackage{setspace}
\usetikzlibrary{shapes,arrows,patterns, decorations.pathreplacing}
\usepackage[colorinlistoftodos,bordercolor=orange,backgroundcolor=orange!20,linecolor=orange,textsize=scriptsize]{todonotes}
\newtheorem{prop}{Proposition}[section]
\newtheorem{thm}[prop]{Theorem}

\newtheorem{question}[prop]{Question}
\newtheorem{cor}[prop]{Corollary}

\newtheorem{problem}{Problem}
\theoremstyle{definition}
\newtheorem{definition}{Definition}
\theoremstyle{remark}
\newtheorem{rmk}[prop]{Remark}

\newtheorem{eg}{Example}
\newcommand{\R}{\mathbb{R}}
\newcommand{\Z}{\mathbb{Z}}
\newcommand{\N}{\mathbb{N}}
\newcommand{\Q}{\mathbb{Q}}
\newcommand{\TP}{\mathbb{TP}}
\newcommand{\sS}{\mathcal{S}}
\newcommand{\barep}{\bar{\varepsilon}}

\newcommand{\tildeep}{\tilde{\varepsilon}}

\newcommand{\hilbertnorm}[1]{\|#1\|_{H}}
\newcommand{\HB}{B^{m-1}_H(\varepsilon)}
\newcommand{\HBd}[1]{B_{H}^{#1}(\delta)}
\newcommand{\HBall}[3]{B_{H}^{#1}(#2, #3)}
\newcommand{\TBall}{B_2^{m-1}(\frac{\varepsilon}{\sqrt{2}})}

\newcommand{\tconv}{\operatorname{tconv}}
\newcommand{\tper}{\operatorname{tper}}

\newcommand{\conv}{\operatorname{conv}}
\newcommand{\troprank}{\operatorname{trank}}
\newcommand{\Vol}{\operatorname{Vol}}
\newcommand{\poly}{\operatorname{poly}}
\newcommand{\type}{\operatorname{type}}

\newcommand{\C}[2]{C^{#1}_{m,n,#2,W}}
\newcommand{\radius}{\operatorname{r}}
\newcommand{\aff}{\operatorname{aff}}
\newcommand{\lift}[1]{\bm{#1}}
\newcommand{\puiseux}{\mathbb{K}}
\newcommand{\val}{\operatorname{val}}

\numberwithin{equation}{section}
\begin{document}
\title{Approximating the Volume of Tropical Polytopes is Difficult}
\author{St\'ephane Gaubert and Marie MacCaig}
\address[A1,A2]{INRIA and CMAP, \'{E}cole polytechnique, CNRS, Universit\'e Paris Saclay,  
91128 Palaiseau Cedex, FRANCE}
\email{stephane.gaubert@inria.com, m.maccaig.maths@outlook.com}
\date{April 28, 2017}
\begin{abstract} We investigate the complexity of counting the number of integer points in tropical polytopes, and the complexity of calculating their volume.  We study the tropical analogue of the outer parallel body and establish bounds for its volume.  We deduce that there is no approximation algorithm of factor $\alpha=2^{\poly(m,n)}$ for the volume of a tropical polytope given by $n$ vertices in a space of dimension $m$, unless P$=$NP.  Neither is there such an approximation algorithm for counting the number of integer points in tropical polytopes described by vertices.  If follows that approximating these values for tropical polytopes is more difficult than for classical polytopes.  Our proofs use a reduction from the problem of calculating the tropical rank.  For tropical polytopes described by inequalities we prove that counting the number of integer points and calculating the volume are $\#$P-hard.

\end{abstract}
\maketitle

\section{Introduction}

Tropical polytopes have appeared in a number of different areas,
including algebraic combinatorics~\cite{TropConv,TropRank,JSY07},
discrete event systems~\cite{ccggq99,katz07,uli2013}, 
scheduling~\cite{MSS04}, program verification~\cite{AGG08},
or game theory~\cite{AGG}. They provide a fundamental
class of modules or convex cones over the tropical semiring~\cite{cgq02}.
They are special instances of tropical convex sets, which arise as log-limits of classical convex sets~\cite{BriecHorvath04}. Tropical polytopes coincide with images by the nonarchimedean valuation of polytopes over non-archimedean fields~\cite{develin_yu,benchimol2013}. We refer the reader to the monographs~\cite{BCOQ92,butkovic,maclagan_sturmfels} for more background on tropical algebra. 

Computing the volume of classical polytopes is a central problem in discrete geometry, as is the question of counting the number of integer points. In this paper, we study the analogous questions for tropical polytopes. 

Recall that computing the volume of classical polytopes is $\#$-P-hard regardless of whether the polytope is described by oracles \cite{Elekes}, inequalities or vertices \cite{DyerFrieze}.  Approximating the volume of a classical polytope is also hard \cite{BaranyFuredi} but a random polynomial time approximation algorithm does exist~\cite{RandomAlg}.

Classically, counting integer points is $\#$P-hard for integral polytopes described by vertices or inequalities (Chapter 7 \cite{DiscreteCompGeometry}).  For polytopes $P=\{x: Ax\leq b\}$, even when restricting the matrix $A$ to the subclass of unimodular matrices which describe contingency tables the problem remains $\#$P-hard \cite{DKM97}, see also \cite{LS03,Mount00}.  Barvinok's algorithm for counting integer points is polynomial for fixed dimension \cite{Bar94}.  For more detail on counting integer points, see the survey \cite{Loera05a}.  

Tropical polytopes are known to be polyhedral complexes whose cells, which are both classically and tropically convex, are called polytropes~\cite{TropConv}.
In this way, the Euclidean volume of a tropical polytope is well defined as the sum of the volumes of its different convex cells. In the aforementioned applications, tropical polytopes typically represent feasible sets.
We are interested
in attaching a measure of importance to these sets, and the Euclidean
volume appears to be the most natural measure. 
For instance, it follows from known results~\cite{TropRank,AGG,Grigoriev}
(see Subsection \ref{subsec:details.AGG.G.DSS} for details)
that determining whether a tropical polytope has zero Euclidean volume is equivalent to deciding whether a mean payoff game is winning.  Hence, the Euclidean volume can play the role of a condition number, measuring how close a game is from not being winning. Moreover, one of our main motivations is to count the number of integer points of tropical polytopes. 
Integer points in the tropical space arise in particular when considering
images by a {\em discrete} nonarchimedean valuation of polytopes over
nonarchimedean fields (like Laurent series), see e.g.~\cite{JSY07}.
As for classical polytopes, the Euclidean volume of a tropical polytope determines the asymptotic behaviour of the number of integer points of a family of scaled polytopes, hence, the Euclidean volume appears to be the canonical notion from this perspective. The reader may note, however, that there are interesting non-classical notions of volumes for tropical polytopes.
These notions have very different properties from the complexity theoretical point of view~\cite{DGJ16}.
Unless otherwise specified, the term ``volume'' will always refer to the Euclidean volume in the sequel. 


We show in this paper that, for tropical polytopes described by a list of vertices, it is hard to approximate either the number of integer points or the volume.  For tropical polytopes described by inequalities, both these problems are proved to be $\#$P-hard.  More precisely, our main results are the following theorems. 

\begin{thm}\label{thm:hardnessApproxVol}  There is no polynomial time approximation algorithm of factor $\alpha=2^{\poly(m,n)}$ for the volume of a full dimensional tropical polytope given by $n$ vertices in a space of dimension $m$ provided P$\neq$NP. 
\end{thm}

\begin{thm}\label{thm:HardnessApproximateCounting}  There is no polynomial time approximation algorithm of factor $\alpha=2^{\poly(m,n)}$ for counting the number of integer points in a full dimensional tropical polytope given by $n$ vertices in a space of dimension $m$ provided P$\neq$NP. 
\end{thm}

The above two theorems demonstrate that these problems are more difficult for tropical polytopes than they are for classical polytopes.  Indeed, in the classical case, a suitably defined convex body $K$ can be bounded by scalings of an ellipsoid $E$ (Theorem 4.6.1 of \cite{GLS93}): $E(\gamma, a)\subseteq K\subseteq E(\Gamma, a)\text{ where } {\Gamma}/{\gamma}=n(n+1)^2.$  It follows that there is a polynomial time deterministic approximation algorithm of factor $(n(n+1)^2)^n$ for the volume of a facet defined classical polytope. This major discrepancy between complexity bounds in the tropical and in the classical setting can be understood intuitively by noting that tropical polytopes can be encoded, loosely speaking, by classical polytopes whose generators have large coefficients (with a number of bits exponential in the size of the input)~\cite{1405.4161}. Hence, it seems
to be a general principle so far that results of computational complexity
in the Turing (bit) model of computation cannot be transferred from the classical setting to the tropical setting.

It should be noted that deciding whether a tropical polytope has zero volume is simpler than approximating the volume.
Indeed, as mentioned above, deciding the former property is equivalent  to deciding whether a mean payoff game is winning. 
Mean payoff games are a well known problem, with an unsettled complexity. They are not known to be in P, but they are in NP $\cap$ co-NP. Hence, deciding if a tropical polytope has zero volume is not NP-hard unless NP$=$co-NP. 

Our proof of the main theorems relies on the following construction, of independent interest. Recall that if $P$ is a classical convex body, the outer parallel bodies of $P$ are defined as the Minkowski sums of $P$ with Euclidean balls.  If $P$ is a tropical polytope, for each $\varepsilon>0$, we define the {\em Hilbert outer parallel body} of $P$ to be the usual Minkowski sum $P+B_H(\varepsilon)$,
where $B_H(\varepsilon)$ is a ball of radius $\varepsilon$ in Hilbert's seminorm; $\hilbertnorm{x}=\max_{i=1,\dots,n} x_i- \min_{i=1,\dots,n} x_i.$ The latter seminorm is related to Hilbert's projective metric, it was shown in~\cite{cgq02} to 
be a canonical metric in tropical convexity (for instance, the classical best approximation results in Euclidean geometry carry over to this metric). We shall see that the Hilbert outer parallel body of a tropical polytope is still a tropical polytope. 

We recall that the tropical rank determines the dimension of a tropical polytope, thought of as a polyhedral complex~\cite{TropRank}.
The main proof ingredient of Theorem~\ref{thm:hardnessApproxVol} consists of the metric estimates in Theorem~\ref{thm:VolBounds},
which imply that the tropical rank $k$ of a polytope can be recovered from the leading term $\varepsilon^{m-k}$ in the asymptotic expansion of the volume of $P+B_H(\varepsilon)$ as $\varepsilon\to 0$.  
These estimates have an intuituive geometric interpretation.
The lower bound involves the notion of {\em inner radius} $r$,
which was introduced by Sergeev in~\cite{DefiniteClosures}, and represents
the maximal radius of a Hilbert ball included
in a polytrope. 
We shall see that every maximal dimensional cell of the polyhedral 
complex defined by $P$ is tropically isomorphic to a full dimensional
cell (polytrope) $X_T$ in a lower dimensional space: the inner radius
of any cell of the latter form arises in our lower bound.
The upper bound involves 
the number $R$ which is a kind of ``outer radius''
of $P$. Thus, the ratio $R/r(X_T)$ measures how far the polytope is
from being ``rotund''. 
More details on the notation can be found in Sections \ref{sec:Hardness} and \ref{sec:HardnessCountingProof}.


\begin{thm}\label{thm:VolBounds} Let $A=(a_{ij})\in\Z^{m\times n}$.  Suppose that the tropical rank of $A$ is $k$ and that $a_{1j}=0, \forall j\in\{1,\dots,n\}$. Let $P=\tconv(A)$ be the tropical polytope
generated by the columns of $A$.  Then, for all $0<\varepsilon\leq R\sqrt{m-1}$, 
\[ 
{{2}}^{\frac{k-m}{2}}\kappa_{m-k} k
\radius(X_T)^{k-1}\varepsilon^{m-k}\leq \Vol^{m-1}(P+B^{m-1}_H(\varepsilon))\leq 
2^{m+k-1}3^{m+n-2}(m-1)^{\frac{k-1}{2}}
R^{k-1}
\varepsilon^{m-k}
\] where $\kappa_{m-k}$ is the volume of the $(m-k)$-dimensional unit ball,  $X_T$ is a cell tropically isomorphic to a cell of maximal dimension
of the polyhedral complex of $P$, and
 $$R:= \max_{i=2,\dots,m}\hilbertnorm{A_{i\cdot}}.$$\end{thm}

We deduce from this theorem that a polynomial-time algorithm for approximating the tropical volume would allow us, by applying this algorithm 
to the tropical polytope $P+B^{m-1}_H(\varepsilon)$,
to compute the tropical rank in polynomial time, whereas it was proved by Kim and Roush \cite{KR05} and Shitov \cite{Sh13} that computing the tropical rank is NP-hard.  The validity of our
reduction relies on the observation that the combinatorial terms appearing as a factor of 
$\radius(X_T)^{k-1}\varepsilon^{m-k}$ or 
$R^{k-1} \varepsilon^{m-k}$ in the latter inequalities have
logarithms that are polynomially bounded in the size of the input. 
Finding the optimal factors is probably difficult,
the merit of these ones is to  be explicit and to allow
relatively short proofs.

In a somewhat similar manner, the proof of Theorem~\ref{thm:HardnessApproximateCounting} relies on the following bounds for the number of integer points of a scaled tropical polytope. 
\begin{thm}\label{thm:CountingBounds} Let $A\in\Z^{m\times n}$, $P=\tconv(A)$, 
and let $k$ denote the tropical rank of $A$.  Then, for any $s\in\mathbb{N}$, $$(\lfloor s\radius(X_T)\rfloor+1)^{k}-\lfloor s\radius(X_T)\rfloor^k\le|sP\cap\Z^{m-1}|\leq 3^{m+n-2}\frac{(1+2sR\sqrt{m-1})^{k-1}-1}{2sR\sqrt{m-1}}$$
where $\radius(X_T)$
is defined as in Theorem~\ref{thm:VolBounds}.
\end{thm}

The paper is organised as follows.  Section \ref{sec:Prelim} contains the required background for our results.  We give simple formula for the volume and number of integer points in Hilbert Balls in Section \ref{sec:HilbertBalls}, which we extend to give lower bounds on these values for polytropes by considering inscribed Hilbert balls.  
Section \ref{sec:ProofBounds} is dedicated to proving bounds on the volume of the Hilbert outer parallel body (Theorem \ref{thm:VolBounds}), which are then used to prove Theorem \ref{thm:hardnessApproxVol} in Section \ref{sec:Hardness}.
We prove bounds on the number of integer points in the dilation of a tropical polytope (Theorem \ref{thm:CountingBounds}) in Section \ref{sec:CountingIntegerPoints} and prove hardness of approximating this value in Section \ref{sec:HardnessCountingProof}.  

In the above results, the tropical polytope is described by generators.  We may ask what the present results become when the polytope is described by inequalities.  We prove in Section \ref{sec:SharpPforInequalities} that, in the latter case, computing the volume or counting the number of integer points is $\#$P-hard.

We finally point out some open problems in Section~\ref{sec-comments}.

\section{Definitions and preliminary results} \label{sec:Prelim}
\subsection{Tropical polytopes}
We work in the tropical semifield $(\R_{\min}, \oplus, \odot)$ where $\R_{\min}:=\R\cup\{\infty\}$,  $a\oplus b:=\min(a,b)$, and the multiplication operation is, $a\odot b:=a+b$.  The pair $(\oplus, \odot)$ extends to matrices and vectors in the same way as in linear algebra, that is (assuming compatibility of sizes)\begin{align*} (A\oplus B)_{ij}&=a_{ij}\oplus b_{ij},\ (A\odot B)_{ij}=\bigoplus_k a_{ik}\odot b_{kj}\text{ and }(\alpha\odot A)_{ij}=\alpha\odot a_{ij}.\end{align*} 

Throughout we use $\oplus$ and $\odot $ to denote the operations in tropical algebra. The notation $+$ denotes classical addition and the Minkowski sum.  Finally $\cdot$ denotes classical multiplication and $ab$ (for $a, b\in\R$) will always mean $a\cdot b$.
   
A subset $S$ of $\R^n$ is a  \emph{tropical convex cone} if $$a\odot x\oplus b\odot y\in S \ \forall x,y\in S, \ a,b\in\R.$$  
We shall call these sets \emph{tropically convex}, for brevity.  The smallest tropical convex subset of $\R^m$ containing $S$ is the \emph{tropical convex hull} of $S$, denoted $\tconv(S)$.  In a slight abuse of notation, we will, for a matrix $A\in\R^{m\times n}$, write $\tconv(A)$ to mean $\tconv(\{A_j:j\in N\})$, that is, the tropical convex hull of the vertices described by the columns, $A_j$, of $A$.

Tropically convex sets can be identified with their image in the $(m-1)$ dimensional \emph{tropical projective space} $$\TP^{m-1}=\R^m/(1,\dots,1)\R.$$  
We may identify an equivalence class $x\in\TP^{m-1}$ to its unique
representative in the space 
$H_j:= \{x\in \R^m\mid x_j =0\}$ where $1\leq j\leq m$ is a fixed
index. By default, we will choose $j=1$. Then, we may represent the
equivalence class of $(0,x_2,\dots,x_m)^T$
by the point $(x_2, x_3, \dots, x_m)^T$. In this way, subsets
of $\TP^{m-1}$ may be visualised as subsets of $\R^{m-1}$.
 We refer the reader to \cite{cgq02,TropConv} for more background on tropical convexity. 

A \emph{tropical polytope} is the tropical convex hull of a finite subset $S$.  Every tropical polytope in $\mathbb{TP}^{m-1}$ is the support of a polyhedral complex, which is described by Develin and Sturmfels in~\cite{TropConv}. We next recall some properties of this complex.

\begin{prop}[{\cite{TropConv}}]\label{prop:PolytopeCellSummary} Given is a tropical polytope $P=\tconv(A)\subseteq\TP^{m-1}$ where $A\in\R^{m\times n}$

1.  Each point $x\in\TP^{m-1}$ has a type, $\type(x)=(S_1,\dots,S_m)$, where $S_i\subseteq \{1,\dots,n\}$ defined by $$j\in S_i\text{ if } a_{ij}-x_i=\min_{k=1,\dots,m}(a_{kj}-x_k)$$

2.  Let $X_S=\{x\in\TP^{m-1}: S\subseteq \type(x)\}.$  Then $X_S$ is a closed, convex polyhedron in the usual sense, \begin{equation}\label{eqn:CellInequalities} X_S=\{x\in\TP^{m-1}: x_k-x_i\leq a_{kj}-a_{ij} \forall i,k\in\{1,\dots, m\} \text{ s.t. } j\in S_i\}.\end{equation}

3.  Each $X_S$ is called a cell.

4. Each bounded cell $X_S$ is a tropical polytope, and is also convex in the usual sense.

5.  The collection of bounded cells $X_S$ provide a polyhedral decomposition of the tropical polytope $P$, called the tropical complex generated by $A$.
\end{prop}
 
We define $\sS^d$ to be the set of all types of bounded cells of dimension $d$, that is, \begin{equation}\label{eqn:No.CellsDim}\sS^d:=\{S=(S_1,\dots,S_m): X_S \text{ is bounded with dimension } d\}.\end{equation}

Without loss of generality we will assume that the generators of a tropical polytope in $\TP^{m-1}$ are given by a matrix $A\in\R^{m\times n}$ with $a_{1j}=0$ for all $j\in \{1,\dots,m\}.$
This is consistent with representing an element of $\TP^{m-1}$
by a representative $(0,x_2,\dots,x_m)^T\in H_1\simeq \R^{m-1}$.

We define the \emph{dimension} of a tropical polytope in $\mathbb{TP}^{m-1}$ to be the maximal dimension of a convex cell of the associated polyhedral complex.  The zero dimensional cells of a tropical polytope are called its \emph{pseudovertices}.

\begin{prop}[{\cite[Coro.~25]{TropConv}}]\label{prop:number.faces.P}
Let $P$ be a tropical polytope in $\mathbb{TP}^{m-1}$, generated by $n$ points in general position. Then, the number of 
faces of dimension $k$ of the tropical complex associated to $P$ is given by the trinomial coefficient 
$${n+m-k-2 \choose n-k-1, m-k-1, k}.$$
\end{prop}
\begin{rmk}We note here that, if the points are not in general position, then Proposition \ref{prop:number.faces.P} provides an upper bound on the number of bounded cells of each dimension.
\end{rmk}

Given a set $T\subset \R^m$ let $\Vol^m(T)$ denote the usual $m$-dimensional volume (the Lebesgue measure) of $T$.  We define the volume of a tropical polytope in dimension $m-1$ to be the sum of the volumes of its full dimensional
(i.e., ($m-1)$-dimensional) bounded cells; that is, given $A\in\R^{m\times n}$ the polytope $P=\tconv(A)\subseteq\TP^{m-1}$ has volume

$$\Vol^{m-1}(P):=\sum_{S\in\sS^{m-1}} \Vol^{m-1}(X_S),$$
each $X_S$ being identified to a subset of $\R^{m-1}$. 

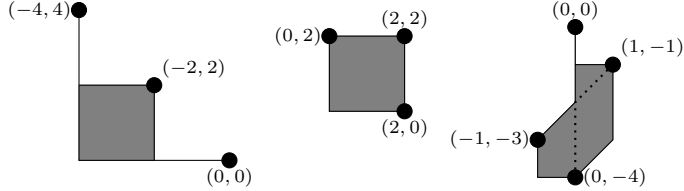
\begin{figure}\begin{center} 
\tiny
\begin{tikzpicture} 
\filldraw[color=gray] (-2,1)--(-1,1)--(-1,0)--(-2,0)--(-2,1);
\draw (0,0)--(-2,0)--(-2,2);
\draw (-1,0)--(-1,1)--(-2,1);
\filldraw (0,0) circle (0.1cm) (-1,1) circle (0.1cm) (-2,2) circle (0.1cm);
\draw (0,0) node[anchor=north]  {$(0,0)$} (-1,1) node[anchor=south west] {$(-2,2)$}  (-2,2) node[anchor=east]  {$(-4,4)$};\end{tikzpicture}
\begin{tikzpicture}
\filldraw[color=gray] (-1,2)--(-1,1)--(0,1)--(0,2)--(-1,2);
\draw (-1,2)--(-1,1)--(0,1)--(0,2)--(-1,2);
\filldraw (0,1) circle (0.1cm) (0,2) circle (0.1cm) (-1,2) circle (0.1cm);
\draw (0,1) node[anchor=north]  {$(2,0)$} (0,2) node[anchor=south] {$(2,2)$}  (-1,2) node[anchor=east]  {$(0,2)$}
(0,0) node {};\end{tikzpicture}
\begin{tikzpicture}
\filldraw[color=gray] (0,0)--(0,-1)--(-0.5, -1.5)--(-0.5,-2)--(0,-2)--(0.5,-1.5)--(0.5,-0.5)--(0,-0.5);
\draw (0,0)--(0,-0.5)--(0,-1)--(-0.5, -1.5)--(-0.5,-2)--(0,-2)--(0.5,-1.5)--(0.5,-0.5)--(0,-0.5);
\filldraw (0,0) circle (0.1cm) (-0.5,-1.5) circle (0.1cm) (0,-2) circle (0.1cm) (0.5,-0.5) circle (0.1cm);
\draw (0,0) node[anchor=south]  {$(0,0)$} (0.5,-0.5) node[anchor=south west] {$(1,-1)$}  (0,-2) node[anchor=west]  {$(0,-4)$} (-0.5, -1.5) node[anchor=east] {$(-1,-3)$};
\draw[thick, dotted] (0.5,-0.5)--(0,-1)--(0,-2);\end{tikzpicture}
\end{center}\caption{\small{Three tropical polytopes in $\TP^2$ (drawn in $\R^2$).}} \label{Figure1}\end{figure}
\normalsize

\begin{eg}In Figure \ref{Figure1} we have drawn three tropical polytopes in $\TP^2$.  The first two are generated by three points in $\R^3$ and have one full dimensional bounded cell.  The last has four generators, and three full dimensional bounded cells.  Each of the tropical polytopes has volume $4$.\end{eg}

\begin{rmk}
The definition of volume of a tropical polytope is independent of which coordinate we choose to fix when identifying $\TP^{m-1}$ to $\R^{m-1}$. Indeed, we may preserve the symmetry between coordinates by identifying
$\TP^{m-1}$ to $H:=\{x\in \R^m\mid \sum_i x_i =0\}$. Then, the absolute value of the determinant of the linear isomorphism $\pi_j: H \to H_j, x\mapsto x-x_j e_j$, where $e_j$ is the $j$th basis vector, can be checked to be independent of the choice of $j$. 
Now, if $P\subset \TP^{m-1}$, we may identify $P$ to a subset $P'$ of $H$.
It follows that $\Vol^{m-1}(\pi_j P')= |\det \pi_j| \Vol^{m-1}(P')$ is independent of the choice of $1\leq j\leq m$. Observe that $\Vol^{m-1}(\pi_1 P')$ coincides
with our initial definition of volume. 
\end{rmk}

\begin{eg}  Let $$A=\begin{pmatrix} -1&-4&-7\\-3&-2&2\\2&-1&-3\end{pmatrix}.$$  Let $P=\tconv(A)$.  In Figure \ref{Figure2} we have drawn each of the polytopes $P\cap H_1$, $P\cap H_2$ and $P\cap H_3$.  Observe that the volume of the tropical polytope is always 4.

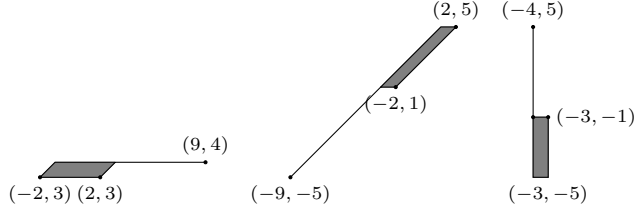
\begin{figure}\begin{center} 
\tiny
\begin{tikzpicture}[scale=0.2]
\filldraw[color=gray] (3,4)--(-1,4)--(-2,3)--(2,3)--(3,4);
\draw (9, 4)--(-1,4)--(-2,3)--(2,3)--(3,4);
\filldraw (9,4) circle (0.1cm) (2,3) circle (0.1cm) (-2,3) circle (0.1cm);
\draw (9,4) node[anchor=south]  {$(9,4)$} (2,3) node[anchor=north] {$(2,3)$}  (-2,3) node[anchor=north]  {$(-2,3)$};
\end{tikzpicture}
\begin{tikzpicture}[scale=0.2]
\filldraw[color=gray] (-3,1)--(1,5)--(2,5)--(-2,1)--(-3,1);
\draw (-9,-5)--(1,5)--(2,5)--(-2,1)--(-3,1);
\filldraw (-9,-5) circle (0.1cm) (-2,1) circle (0.1cm) (2,5) circle (0.1cm);
\draw (-9,-5) node[anchor=north]  {$(-9,-5)$} (-2,1) node[anchor=north] {$(-2,1)$}  (2,5) node[anchor=south]  {$(2,5)$};
\end{tikzpicture}
\begin{tikzpicture}[scale=0.2]
\filldraw[color=gray] (-4,-1)--(-4,-5)--(-3,-5)--(-3,-1)--(-4,-1);
\draw (-4,5)--(-4,-5)--(-3,-5)--(-3,-1)--(-4,-1);
\filldraw (-4,5) circle (0.1cm) (-3,-1) circle (0.1cm) (-4,-1) circle (0.1cm);
\draw (-4,5) node[anchor=south]  {$(-4,5)$} (-3,-1) node[anchor=west] {$(-3,-1)$}  (-3,-5) node[anchor=north]  {$(-3,-5)$};
\end{tikzpicture}
\end{center}\caption{\small{On the left, $P\cap H_1$.  In the center, $P\cap H_2$.  On the right, $P\cap H_3$.}}\label{Figure2} \end{figure}
\normalsize

\end{eg}

\begin{rmk}Whilst we consider in this paper the Euclidean volume of tropical polytopes, we remark that there is another natural definition of the volume of a tropical polytope, arising from dequantisation, as discussed in \cite{DGJ16}.  Every tropical polytope $P=\tconv(\{V_1,\dots,V_n\})$ corresponds to a classical polytope over the field of absolutely convergent Puiseux series, $\lift{P}_t=\conv(\{\lift{V}_1,\dots,\lift{V}_n\})$ through the valuation $\val: \puiseux\rightarrow \R$ which sends $f\in\puiseux$ to its highest/leading exponent.  The authors define the dequantised tropical volume as  $$\lim_{t\to \infty}\frac{\log{ \Vol(\lift{P}_t)}}{\log{t}},$$ and show that, under genericity assumptions, it can be computed in polynomial time.  \end{rmk}

The Hilbert seminorm is given by $$\hilbertnorm{x}=\max_{i=1,\dots,m} x_i- \min_{i=1,\dots,m} x_i.$$  It is the tropical analogue of the Euclidean norm.  The Hilbert ball in dimension $d-1$ (centered at the origin), is \begin{equation}\label{eqn:HilbertGeneratingMatrix}B^{d-1}_H(\varepsilon):=\{ x\in\mathbb{TP}^{d-1}: \max_{i=1,\dots,d} x_i- \min_{i=1,\dots,d}x_i\leq \varepsilon\}\end{equation} where $\varepsilon>0$.  Note that $B^{d-1}_H(\varepsilon)$ is the tropical polytope $\tconv(H)$ for \begin{equation}\label{eqn:MatrixHforHilbertBall} H=\begin{pmatrix}0&&&\\&0&\varepsilon&\\&\varepsilon&\ddots&\\&&&0 \end{pmatrix}\in\R^{d\times d}\end{equation} where all off diagonal entries are $\varepsilon$.  The Hilbert Ball in $\TP^2$ is shown in Figure \ref{Figure3}.

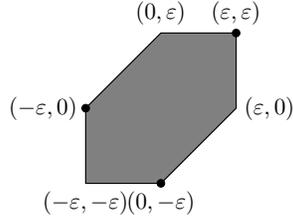
\begin{figure}\begin{center} 
\begin{tikzpicture} [scale=0.5, every node/.style={scale=0.8}]
\filldraw[color=gray] (0,-2)--(2,0)--(2,2)--(0,2)--(-2,0)--(-2,-2)--(0,-2);
\draw (0,-2)--(2,0)--(2,2)--(0,2)--(-2,0)--(-2,-2)--(0,-2);
\draw (0,-2) node[anchor=north] {$(0,-\varepsilon)$} (2,0) node[anchor=west] {$(\varepsilon,0)$} (2,2) node[anchor=south] {$(\varepsilon,\varepsilon)$} (0,2) node[anchor=south] {$(0,\varepsilon)$} (-2,0) node[anchor=east] {$(-\varepsilon,0)$} (-2,-2) node[anchor=north] {$(-\varepsilon,-\varepsilon)$};
\filldraw (-2,0) circle (0.1cm) (0,-2) circle (0.1cm) (2,2) circle (0.1cm);
\end{tikzpicture}
\end{center}\caption{\small{The Hilbert Ball in $\TP^2$.}} \label{Figure3}\end{figure}

Given $A\in\R^{m\times m}$ its \emph{tropical permanent} is $$\tper(A)=\min_{\pi} a_{1\pi(1)}+\dots + a_{m\pi(m)}$$ where the minimum is taken over all permutations $\pi$ of $\{1,\dots, m\}$.  A square matrix $A\in\R^{m\times m}$ is \emph{tropically singular} if the minimum in $\tper(A)$ is attained at least twice. 
The \emph{tropical rank} of a matrix $A$, denoted by $\troprank(A)$, is the largest integer $r$ such that $A$ has a tropically nonsingular $r\times r$ minor.

\begin{thm}[{\cite{TropRank}}]\label{thm:trank.dim}
Let $A\in\R^{m\times n}$. 
The tropical rank of $A$ is equal to $1+d$ where $d$ is the dimension of the polytope $\tconv(A)$ in $\TP^{m-1}$.
\end{thm}

\subsection{The equivalence between zero volume and mean payoff games}\label{subsec:details.AGG.G.DSS}

In this section we recall that zero volume is equivalent to mean payoff games.  This section is self contained and can be skipped without affecting readability of the paper.  

Recall that the column vectors of a matrix $A$ are \emph{tropically dependent} if there exists some $x$ not identically $\infty$,  such that in the expression $\min_{j}(a_{ij}+x_j)$ the minimum is attained at least twice for every $i$.  The columns are called \emph{tropically linearly independent} otherwise.  The following basic result links tropical linear independence with tropical nonsingularity. It was proved by Develin, Santos and Sturmfels in \cite{TropRank} (for matrices with finite entries), by Izhakian and Rowen in~\cite{IR09}, and by Akian, Gaubert and Guterman in~\cite{AGG}.

\begin{thm}[{\cite{TropRank,IR09,AGG}}]\label{prop:LinIndepColandNonsingular}
Let $A\in\R_{\min}^{m\times n}$ with $m\geq n$.  The columns of $A$ are tropically linearly independent if and only if $A$ has an $n\times n$ submatrix which is tropically nonsingular.
\end{thm}

The next theorem follows by combining results of \cite{AGG} and of Grigoriev and Podolskii \cite{Grigoriev}.

\begin{thm}[Corollary of~{\cite{AGG,Grigoriev}}]\label{thm:AGG.G.DSS} The following problems are equivalent with respect to polynomial time Karp (see \cite{Karp}) reductions:

(P1) Determining whether a tropical polytope given by generators with rational entries has zero volume,

(P2) Deciding whether a rectangular matrix $A\in\mathbb{Q}^{m\times n}$ does not have maximum tropical rank,

(P3) Checking whether a mean payoff game is winning.
\end{thm}

\proof

{\em P1 $\Leftrightarrow$ P2:}
Recall that $\sS^d$ denotes the set of types of $d$-dimensional cells of the polyhedral
complex associated to $P=\tconv(A)$, with $A\in\Q^{m\times n}$,
see \eqref{eqn:No.CellsDim}.  Now $\Vol^{m-1}(P)=0\Leftrightarrow \sS^{m-1}=\emptyset.$
Further, $m-1$ dimensional cells correspond to $m\times m$ submatrices of $A$ so, by Theorem \ref{thm:trank.dim}, $(\exists S\in\sS^{m-1})\Leftrightarrow \troprank(A)=m.$  The claim follows.

{\em P2 reduces to P3:}
Tropical linear independence reduces to the absence of a winning strategy in a MPG by Section 4.2 of~\cite{AGG}.  Since checking whether the columns of $A$ are tropically dependent is equivalent to $A$ not having maximum tropical rank by Theorem~\ref{prop:LinIndepColandNonsingular}, it follows that P2 reduces to P3.

{\em P3 reduces to P2:}
This follows from Corollary 8 of \cite{Grigoriev}, using again
the equivalence in Theorem \ref{prop:LinIndepColandNonsingular}.
\endproof

\section{Hilbert Balls}\label{sec:HilbertBalls}

Recall that the Hilbert ball in dimension $d-1$ of radius $\delta$ centered at $v\in\TP^{d-1}$ is $$\HBall{d-1}{v}{\delta}:=\{v+x\in\TP^{d-1}:x_i-x_j\leq\delta, \text{ for all }i,j\in\{1,\dots, d\},\ x_1=0\}.$$  We write $\HBd{d-1}$ for $\HBall{d-1}{0}{\delta}$.
\subsection{The number of integer points and volume of Hilbert balls}

We prove exact formulas for the volume of, and number of integer points in, Hilbert Balls. 
We will use the following result on the Ehrhart polynomial \cite{Ehr62}, see also \cite{BLDPS05}.  

\begin{thm}[{\cite{Ehr62}}]\label{thm:EhrhartPoly} The number of lattice points in an integer scaling of a lattice polytope of dimension $d$ in $\R^m$ is a rational polynomial function of degree $d$.
Further, the leading coefficient of this polynomial is the volume of $P$, \begin{equation}\label{eqn:EhrhartPoly} |sP\cap\Z^m|=\Vol^d(P)s^d+c_{d-1}k^{d-1}+\dots+a_0.\end{equation}  \end{thm}

\begin{prop}\label{prop:HilbertBallIntPoints} For any integer vector $v$ and $\delta>0$, $$|\HBall{d-1}{v}{\delta}\cap\Z^{d-1}|=|\HBd{d-1}\cap\Z^{d-1}|=(\lfloor\delta\rfloor+1)^d-\lfloor\delta\rfloor^d.$$
\end{prop}
\proof The first equality is trivial.  We show $|\HBd{d-1}\cap\Z^{d-1}|=(\lfloor\delta\rfloor+1)^d-\lfloor\delta\rfloor^d$  by considering the function $\pi: [0, \delta]^d\rightarrow \R^d$ given by $\pi(x)=x-x_1e$ where $e=(1, 1,\dots, 1)^T$.

Let $S:=\{x\in[0,\delta]^d: \min_ix_i=0\}$.  We claim that $\pi: S\rightarrow \HBd{d-1}$ is a bijection.  Indeed:

\noindent (i) $(\forall x\in S) \pi(x)\in \HBd{d-1}$ because, if $y=\pi(x)$, then clearly $y_1=0$ and $y_i-y_j=(x_i-x_1)-(x_j-x_1)=x_i-x_j\leq\delta.$

\noindent (ii) $\pi$ is injective because $\pi(x)=\pi(y)\Rightarrow y=x+\alpha e$ for some $\alpha\in\R$ where $\min_ix_i=0=\min_i y_i\Rightarrow \alpha=0$.

\noindent (iii)  $\pi$ is surjective because, for $y\in\HBd{d-1}$,  $\pi(x)=y$ where $x=y+(\min_iy_i)e\in S$. 

Hence, there is a bijection between $S$ and $\HBd{d-1}$, meaning that 
\begin{align*}|\HBd{d-1}\cap\Z^{d-1}|&=|S\cap\Z^d|
\\&=|[0,\delta]^d\cap\Z^d|-|\{x\in[0,\delta]^d: \min_ix_i\neq 0\}\cap\Z^d|
\\&=|[0,\delta]^d\cap\Z^d|-|[1,\delta]^d\cap\Z^d|
\\&=|[0,\lfloor\delta\rfloor]^d\cap\Z^d|-|[1,\lfloor\delta\rfloor]^d\cap\Z^d|
\\&=(\lfloor\delta\rfloor+1)^d-\lfloor\delta\rfloor^d.
\end{align*}
\endproof

\begin{cor}\label{cor:HilbertBallVolume} For $\delta\in\Q$, $\Vol^{d-1}(\HBall{d-1}{v}{\delta})=\Vol^{d-1}(\HBd{d-1})=d\delta^{d+1}.$
\end{cor}

\proof  Let $s\in\N$ satisfy $s\delta\in\N$.  From Theorem \ref{thm:EhrhartPoly}, $$|s\HBd{d-1}\cap\Z^{d-1}|=\Vol^{d-1}(\HBd{d-1})s^{d-1}+\sum_{i=0}^{d-2}c_is^i$$ and further, by Proposition \ref{prop:HilbertBallIntPoints}, $$|s\HBd{d-1}\cap\Z^{d-1}|=(\lfloor s\delta\rfloor+1)^d-\lfloor s\delta\rfloor^d=d(s\delta)^{d-1}+\sum_{i=2}^{d}{{d}\choose{i}}(s\delta)^{d-i}.$$

By equating the leading coefficients we conclude $\Vol^{d-1}(\HBd{d-1})s^{d-1}=d(s\delta)^{d-1}.$
\endproof

\subsection{Inscribed Hilbert balls in tropical polytropes}
\label{sec-innerradius}
For a tropical polytope $P$ we define the \emph{inner radius of $P$} as $$\radius(P):=\max\{\delta:  B_H(u, \delta)\subseteq P\}.$$  
We consider here tropical polytropes which, by \cite{TropConv}, can be considered to be cells in a decomposition described by a tropical polytope $P=\tconv(A)$, and described by \eqref{eqn:CellInequalities}.

To each cell we associate a matrix as follows.  Let $B_S\in\Z^{m\times m}$ have entries $$b_{kj}:=\begin{cases}\min_{i\in S_j}(a_{ki}-a_{ji})&\text{if } i\neq j,\\ \infty&\text{otherwise.}\end{cases}$$  Then $X_S=\{x\in\TP^{m-1}: x_i-x_j\leq (B_S)_{ij}, 1\leq i,j\leq m\}.$

Note that, throughout, matrices denoted with the letter $A$ will describe vertices of tropical polytopes $P=\tconv(A)$, whereas matrices using the letter $B$ refer to an inequality description of the form described above.

We first remind the reader of a few necessary concepts in min-plus spectral theory.  Given a matrix $A\in\R^{m\times m}$ its \emph{mimimum cycle mean}, denoted $\rho_{\min}(A)$, is $$\min_{\text{Cycles }\sigma}\frac{w(\sigma, A)}{l(\sigma)}$$ where $w(\sigma, A)$ denotes the weight of the cycle and $l(\sigma)$ is its length.  Any cycle which attains this minimum is called a \emph{critical cycle} and the critical graph is the digraph containing only nodes and edges from critical cycles.  For cells $X_S$, we may also write $\rho_{\min}(X_S)$ to mean $\rho_{\min}(B_S)$.

Sergeev characterised in~\cite{DefiniteClosures} the inner radius of polytropes.
\begin{thm}[{\cite{DefiniteClosures}}] \label{th-sergeev}
For a polytrope $X_S$, $\radius(X_S)=\rho_{\min}(B_S)$.
\end{thm}

Since it will be needed in the sequel, we provide the following proposition,
from which Theorem~\ref{th-sergeev} can be recovered.
\begin{prop}\label{prop:InscribedHilbertBall} Let $X_S\subseteq\TP^{m-1}$ be full dimensional with  $X_S=\{x\in\TP^{m-1}: x_i-x_j\leq (B_S)_{ij}\}.$  Let $u$ satisfy $B_S\odot u\geq \rho_{\min}(B_S)\odot u$.  Then, for all $0<\delta\leq\rho_{\min}(X_S)$, $$B_H(u, \delta)\subseteq X_S$$
\end{prop}
\proof  
Fix some $\delta\leq \rho_{\min}(X_S)$.  Write $B=B_S$.  We show $B^{m-1}_H(u, \delta)\subseteq \{x: x_i-x_j\leq B_{ij}\}(=X_S).$
Let $u+h$ be an arbitrary point in $B^{m-1}_H(u, \delta)=\{u+h: |h_i-h_j|\leq \delta\}.$
By our assumptions on $u$, 
$(\forall i\in M)(\forall j\in M)\ B_{ij}-\delta\geq u_i-u_j$ and hence $$(\forall i\in M)(\forall j\in M) \ (u_i+h_i)-(u_j+h_j)= (u_i-u_j)+(h_i-h_j)\leq (B_{ij}-\delta)+\delta=b_{ij}.$$
We conclude that $u+h\in X_S$.
\endproof

\subsection{Consequences of these results on polytropes}

Here we give lower bounds for the volume of, and number of integer points in, polytropes by considering the largest inscribed Hilbert ball.  For full dimensional polytropes the result follows immediately from Propositions \ref{prop:HilbertBallIntPoints}, \ref{prop:InscribedHilbertBall}, Theorem \ref{th-sergeev} and Corollary \ref{cor:HilbertBallVolume}:

\begin{cor}\label{prop:VolandIntBoundsFullDim} Let $X_S\subseteq\TP^{m-1}$ be a bounded polytrope of dimension $m-1$.  Then 

(1)  $\Vol^{m-1}(X_S)\geq k\radius(X_S)^{k-1}$, 

(2)  $|X_S\cap\Z^{m-1}|\geq (\lfloor\radius(X_S)\rfloor+1)^m-\lfloor\radius(X_S)\rfloor^m$.
\end{cor}

The following result tells us that, in the full dimensional case, $\radius(X_S)$ is always positive.

\begin{prop}\label{prop:MCMofFullDimensionalCell} Suppose $A\in\R^{m\times n}$, $P=\tconv(A)$, $\troprank(A)=m$ and let $X_S$ be a full dimensional, bounded cell.  Then $\radius(X_S)>\frac{1}{m}$, and, in fact, $\radius(X_S)={z}/{t}$ where $z\in\N$ and $t\in\{1,\dots,m\}$.
\end{prop}
\proof By the assumption on the rank of $A$ we know that $X_S$ has dimension $m-1$.  By Theorem \ref{th-sergeev}, $\radius(X_S)=\rho_{\min}(B_S)$ so we prove the results for the minimum cycle mean.  We write $B=B_S$.

Now $\rho_{\min}(B)=\bigoplus_{\sigma}w(\sigma, B)\odot l(\sigma)^{-1}$ where $l(\sigma)$ is the length of the cycle $\sigma$.  Let $\sigma$ be a cycle attaining the minimum cycle mean.  Since $B\in\Z^{m\times m}$, clearly $w(\sigma, B)\in\Z$.  Further, it is known that it is sufficient to consider elementary cycles when calculating the minimum cycle mean.  Hence $l(\sigma)\in\{1,\dots,m\}$.

It remains to show that $w(\sigma, B)>0$.  Assume first, for a contradiction, that $\sigma=(i_1,\dots, i_r)$ has weight $b_{i_1i_2}+\dots+b_{i_{r-1}i_r}+b_{i_ri_1}< 0$.  Then, by definition of $B$, \begin{align*}&\sum_{k=1}^{r-1} (x_{i_k}-x_{i_{k+1}})\leq\sum_{k=1}^{r-1} b_{i_ki_{k+1}} \\
\therefore \ & x_{i_1}-x_{i_r}\leq b_{i_1i_2}+\dots+b_{i_{r-1}i_r}< -b_{i_ri_1}
\end{align*}

We conclude, $b_{i_ri_1}< x_{i_r}-x_{i_1}\leq b_{i_ri_1}$ which is impossible.  

If, instead, $w(\sigma, B)=0$ then an almost identical argument shows $b_{i_ri_1}\leq x_{i_r}-x_{i_1}\leq b_{i_ri_1}$ which implies $x_{i_r}=x_{i_1}$ which means that $X_S$ cannot have full dimension, a contradiction.
\endproof

Observe, when $\troprank(A)=k<m$, it is possible to have $\radius(X_S)=0$ for cells $X_S$ in $P$.  So the above argument is only valid for full dimensional polytopes. For polytropes $X_S\subseteq\TP^{m-1}$ of dimension $k-1$, $k<m$, we do not have a formula for the radius of the largest inscribed Hilbert ball.  Instead we identify $X_S$ with a full dimensional cell $X_T\subseteq\TP^{k-1}$ to which it is tropically isomorphic.

\begin{prop}\label{prop:VolandIntPointsBoundXSbyXT}  Let $X_S\subseteq\TP^{m-1}$ be a bounded polytrope of dimension $k-1$, where $k<m$.  Suppose $X_S$ is described by $X_S=\{x\in\TP^{m-1}: x_i-x_j\leq (B_S)_{ij}\}$ and let $D\subset\{1,\dots,m\}$ be any set containing exactly one index from each strongly connected component of the critical graph of $B_S$, so $|D|=k$.  Further, let \begin{equation}\label{eqn:XT} X_T:=\{x\in\TP^{k-1}:  x_i-x_j\leq (B_S)_{ij} \forall i,j\in D\}\subseteq\TP^{k-1}.\end{equation} Then

(1)  $\Vol^{k-1}(X_S)\geq \Vol^{k-1}(X_T)$ and

(2) For any $s\in\N$,
$|sX_S\cap\Z^{m-1}|=|sX_T\cap\Z^{m-1}|$
\end{prop}

\proof 

It follows from Theorems 3.100 and 3.101 in \cite{BCOQ92} that the map $$\psi: (x_{\alpha})_{\alpha\in D}\rightarrow \bigoplus_{d\in D} B^*_{,d}\odot x_d$$ is a tropical isomorphism from $X_T$ to $X_S$ satisfying $\psi_d(x)=x_d, d\in D$.  Hence, after reordering the coordinates, \begin{equation}\label{eqn:PsiMap}\psi(y)=(y, \phi(y)).\end{equation}

(1) By the formula of change of variables, the infinitesimal element of volume, under the action of $\psi$, is multiplied by $$C=\sqrt{\sum_{J}(\det(D\psi)_J)^2}$$ where $D\psi$ denotes the differential map of $\psi$, $(D\psi)_J$ is the maximal minor of $D\psi$ with indices from $J$, and the sum is taken over all maximal minors, see Section 8.72 of \cite{ShilovLA}.

By \eqref{eqn:PsiMap} $D\psi$ has the identity matrix as a maximal submatrix.  It follows that $C\geq 1$ and hence $\Vol^{k-1}(X_S)=C \Vol^{k-1}(X_T)\geq \Vol^{k-1}(X_T).$

(2)   Since $B$ is an integer matrix, it is clear that $\psi$ preserves integer vectors, and is unaffected by an integer scaling, hence $|sX_S\cap\Z^{m-1}|=|sX_T\cap\Z^{k-1}|.$
\endproof

An immediate consequence of Propositions \ref{prop:VolandIntBoundsFullDim} and \ref{prop:VolandIntPointsBoundXSbyXT} is the following.

\begin{cor}\label{cor:VolandIntPointsX_S(usingX_T)} Let $X_S\subseteq\TP^{m-1}$ be a bounded polytrope of dimension $k-1$, where $k<m$.  Suppose $X_S$ is described by $X_S=\{x\in\TP^{m-1}: x_i-x_j\leq (B_S)_{ij}\}.$  Further, let $X_T\subseteq\TP^{k-1}$  be defined by \eqref{eqn:XT}.  Then 

(1)  $\Vol^{k-1}(X_S)\geq  k\radius(X_T)^{k-1}$

(2) $(\forall s\in\N)\ |sX_S\cap\Z^{m-1}|\geq (\lfloor\radius(sX_T)\rfloor+1)^k-\lfloor \radius(sX_T)\rfloor^k.$
\end{cor}

\section{Bounding the Volume of the Hilbert Outer Parallel Body} \label{sec:ProofBounds} 

\subsection{The Hilbert outer parallel body of a tropical polytope}

The proof of our complexity result requires us to generate a full dimensional polytope from any matrix $A$.  Thus we introduce the Hilbert outer parallel body, the tropical analogue of the outer parallel body in classical geometry.

Recall that the \emph{Minkowski sum} of two sets of vectors $V_1$ and $V_2$ is $V_1+V_2:=\{v_1+v_2: v_1\in V_1, v_2\in V_2\}.$   Given a (classical) polytope $P$, the (closed) Euclidean unit ball $B^d$ and $\lambda\in\R$, the polytope $P+ \lambda B^d$ is called the \emph{outer parallel body} of $P$. Similarly, 
we make the following definition, replacing the Euclidean norm
by Hilbert's seminorm.
\begin{definition}[Hilbert outer parallel body]
Given a tropical polytope $P=\tconv(A)\subset \TP^{m-1}$, 
we define the \emph{Hilbert outer parallel body}
of parameter $\varepsilon$ to be the Minkowski sum
\[ P+B_H^{m-1}(\varepsilon) \enspace .
\]
\end{definition}
We note that the idea of considering the Minkowski sum of a Hilbert Ball and
of a tropical polytope was used in \cite[\S~5]{AK13} with a different purpose (studying the external representations of tropical polytopes).

We prove in this section Theorem \ref{thm:VolBounds}, that is, we establish the following upper and lower bounds on the volume of the Hilbert outer parallel body of a tropical polytope $P=\tconv(A)\subseteq\TP^{m-1}$ described by $A\in\Z^{m\times n}$ with $\troprank(A)=k$: $$\bigg(\frac{\varepsilon}{\sqrt{2}}\bigg)^{m-k}\kappa_{m-k} k\radius(X_T)^{k-1}\leq \Vol^{m-1}(P+B^{m-1}_H(\varepsilon))\leq 2^m3^{m+n-2}(2R\sqrt{m-1})^{k-1}\varepsilon^{m-k}.$$
The term $\varepsilon^{m-k}$ here is intuitive: $P$ is a polyhedral complex in $\R^{m-1}$ whose cells have dimension at most $k-1$, and so, the volume of the parallel body is expected to be of order $\varepsilon^{(m-1)-(k-1)}=\varepsilon^{m-k}$. 
The factors of $\varepsilon^{m-k}$ have closed form expressions allowing
relatively short proofs, they are not optimal. For the present application,
we only need to know that the logarithm of these factors are polynomially
bounded in the size of the input. 

We begin by showing that the Hilbert outer parallel body of a tropical polytope is also a tropical polytope.





\begin{prop}\label{prop:HOPB.Generators} Given $P=\tconv(A)$ and $\HB=\tconv(H)$ where $A\in\R^{m\times n}$ and $H\in\R^{m\times m}$ is described by \eqref{eqn:MatrixHforHilbertBall}, a (possibly non-minimal) family of generators of $P+\HB$ is $$\{A_j+H_k: j=1,\dots,n, k=1,\dots,m \}.$$
\end{prop}
\proof Let $x\in P+\HB$.  Then, there exists $\alpha\in\R^n, \beta\in\R^m$ such that \begin{align*}x&=\Big(\bigoplus_j \alpha_j\odot A_j\Big) +\Big(\bigoplus_k \beta_k\odot H_k\Big)=\bigoplus_j\bigoplus_k \alpha_j\odot \beta_k\odot (A_j+H_k).\end{align*}  The claim follows.
\endproof

As a consequence of Proposition \ref{prop:HOPB.Generators} we obtain:
\begin{cor}\label{prop:Outer.Body.Poloytope} If $P$ is a tropical polytope, then so is the Hilbert outer parallel body, $P+B_H^{m-1}(\varepsilon).$
\end{cor}



\begin{eg}
The Hilbert outer parallel body of $P=\tconv(A)\subseteq\TP^2$ is drawn in Figure \ref{Fig:HOPB} where $$A=\begin{pmatrix}0&0&0\\0&-1&-2\\0&1&2 \end{pmatrix}.$$

\begin{figure}
\begin{center}
\begin{tikzpicture}[scale=0.5, every node/.style={scale=0.5}]
\node at (0,-1) {};
\filldraw[color=gray!20!white] (-8,4)--(-4,4)--(-4,0)--(-8,0);
\draw[thick] (0,0)--(-8,0)--(-8,8) (-8,4)--(-4,4)--(-4,0);
\filldraw (0,0) circle (0.1cm) (-4,4) circle (0.1cm) (-8,8) circle (0.1cm);
\end{tikzpicture}\ \ \ \ \ \ \ 
\begin{tikzpicture}[scale=0.5, every node/.style={scale=0.5}]

\filldraw[color=gray!20!white] (-1,-1)--(-9,-1)--(-9,7)--(-9,8)--(-8,9)--(-7,9)--(-7,8)--(-7,7)--(-7,5)--(-3,5)--(-3,1)--(1,1)--(1,0)--(0,-1)--(-1,-1);
\filldraw[pattern=north west lines] (0,-1)--(1,0)--(1,1)--(0,1)--(-1,0)--(-1,-1)--(0,-1);\draw (0,-1)--(1,0)--(1,1)--(0,1)--(-1,0)--(-1,-1)--(0,-1);
\filldraw[pattern=north west lines] (-8,-1)--(-7,0)--(-7,1)--(-8,1)--(-9,0)--(-9,-1)--(-8,-1);\draw (-8,-1)--(-7,0)--(-7,1)--(-8,1)--(-9,0)--(-9,-1)--(-8,-1);
\filldraw[pattern=north west lines] (-8,7)--(-7,8)--(-7,9)--(-8,9)--(-9,8)--(-9,7)--(-8,7);\draw (-8,7)--(-7,8)--(-7,9)--(-8,9)--(-9,8)--(-9,7)--(-8,7);
\filldraw[pattern=north west lines] (-4,3)--(-3,4)--(-3,5)--(-4,5)--(-5,4)--(-5,3)--(-4,3);\draw (-4,3)--(-3,4)--(-3,5)--(-4,5)--(-5,4)--(-5,3)--(-4,3);
\draw (-1,-1)--(-9,-1)--(-9,7) (-7,8)--(-7,7)--(-7,5)--(-3,5)--(-3,1)--(0,1);
\draw[very thick, dotted] (0,0)--(-8,0)--(-8,8) (-8,4)--(-4,4)--(-4,0);
\filldraw (0,0) circle (0.1cm) (-4,4) circle (0.1cm) (-8,8) circle (0.1cm);
\end{tikzpicture}
\end{center}\caption{A tropical polytope (left) and its Hilbert outer parallel body (right).}
\label{Fig:HOPB}\end{figure}
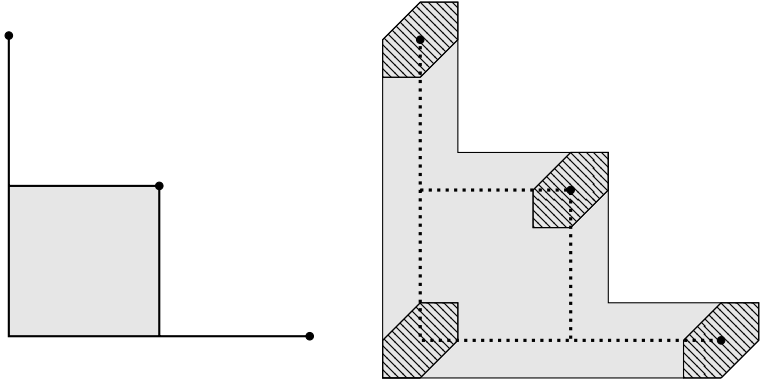
\end{eg}
 
\subsection{Integer tropical polytopes and bounds on number of cells}

An \emph{integer tropical polytope} is a tropical polytope $P$ whose every pseudovertex is integer. 

\begin{prop} \label{prop.intpseudovertices}
Let $A\in\Z^{m\times n}$. Let $P=\tconv(A)$ be a polytope in $\mathbb{TP}^{m-1}$. Then every cell of $P$ has integer pseudovertices, and so, 
$P$ is an integer tropical polytope.\end{prop}
\begin{proof}
Consider a cell $X=X_S$, $S=(S_1,\dots,S_m)$ of $P$.  From \cite{TropConv}, see also Proposition \ref{prop:PolytopeCellSummary}, we have 
$$X=\{x\in\TP^{m-1}: x_k-x_j\leq a_{ki}-a_{ji}\ \forall j,k\in\{1,\dots,m\} \text{ s.t. } i\in S_j\}$$ and that $X$ is convex in the usual sense.  Note that $a_{ij}\in\Z$ for all $i,j$.  Thus $X=\{x: Mx\leq b\}$, where $b$ is an integer vector and $M\in\{0, 1, -1\}^{m-1\times p}$, some $p\in \mathbb{N}$.  
It is clear that $M$ is totally unimodular and hence $X$ has integer vertices.
 \end{proof}

We will require a (very coarse) bound on the number of bounded cells in $P$ which is a consequence of Proposition \ref{prop:number.faces.P}.

\begin{cor}\label{prop:bound.cells.P}
Let $P$ be a tropical polytope in $\mathbb{TP}^{m-1}$, generated by $r$ points in general position. Then $|\sS^t|\leq 3^{r+m-t-2}.$
\end{cor}

\begin{proof} Using Proposition \ref{prop:number.faces.P}, the number of faces/bounded cells of $P$ of dimension $t$ is $${r+m-t-2 \choose r-t-1, m-t-1, t}\leq (1+1+1)^{r+m-t-2}=3^{r+m-t-2}.$$ 
\end{proof}

\subsection{Upper Bound}
We now prove the upper bound in Theorem \ref{thm:VolBounds}.  To do this we first introduce an ``outer radius'' of a tropical polytope, which is the length of sides of the smallest hypercube which contains it.  We remark here that we could bound the volume of a tropical polytope by a smallest circumscribed Hilbert ball, but in many cases a tighter bound is achieved using hypercubes.  While, to our knowledge, no simple formula for the circumradius of a tropical polytope exists, a simple formula for the radius of \emph{some} Hilbert ball containing a tropical polytope is defined in \cite{KleeneStarPolytropes}.

Recall that we assume without loss of generality that $A\in\R^{m\times n}$ has its first row, $A_{1\cdot}$, equal to zero.  Define \begin{equation}\label{eqn:DefR}R:= \max_{i=2,\dots,m}\hilbertnorm{A_{i\cdot}}= \max_{i=2,\dots,m}(\max_{j=1,\dots,n} (a_{ij}-a_{1j})-\min_{j=1,\dots,n}(a_{ij}-a_{1j})). \end{equation}  Note that $\hilbertnorm{A_{i\cdot}}$ is a lower bound on the length of the side (perpendicular to $x_i=0$) of an $(m-1)$-cube containing $P$.  Hence it can be checked (see also Theorem 3.34 in \cite{MacCaigThesis}) that there exists a vector $u$ such that $P=\tconv(A)\subseteq u+[0,R]^{m-1}.$

However, since $P$ has dimension $k-1$, $P$ may be a lower dimensional tropical polytope embedded in the $(m-1)$-dimensional space. So, more work is required to bound the $(k-1)$-volume of $P$.  In what follows, for any cell $X_S$ of $P$ we construct a simple shape whose volume is an upper bound on the volume of $X_S$.  Note that this bound is not optimal, but is constructed for ease of calculations later.

\begin{prop}\label{prop:CellsinCube} Let $A\in\Z^{m\times n}$. Let $P=\tconv(A)$ be the polytope in $\mathbb{TP}^{m-1}$ generated by the columns of $A$.
Then, for all $t\in\{1,\dots,\troprank(A)-1\}$ every cell $X_S$
of $P$ such that $S\in\sS^t$ is contained in a translation and/or rotation of the $t$-cube $[-R\sqrt{m-1}, R\sqrt{m-1}]^{t}$.
Further, $\Vol^{t}(X_S)\leq \Vol^{t}([-R\sqrt{m-1}, R\sqrt{m-1}]^{t})$.
\end{prop}

\begin{proof}
It follows from the above discussion that the largest distance between any two points in $P$ is at most the largest distance between any two points in $[0, R]^{m-1}$, which is $\sqrt{(m-1)R^2}=R\sqrt{m-1}.$
Hence, for any $X_S, S\in\sS^t$, there exists a vector $v$ such that $X_S\subseteq v+B^t_2(0, R\sqrt{m-1}),$ and $\Vol^{t}(X_S)\leq \Vol^{t}(B^t_2(0, R\sqrt{m-1}))$.

Further relaxing this bound we use that $B^t_2(0, R\sqrt{m-1})\subseteq [-R\sqrt{m-1}, R\sqrt{m-1}]^{t}$ to obtain $X_S\subseteq v+[-R\sqrt{m-1}, R\sqrt{m-1}]^{t}$ and $\Vol^{t}(X_S)\leq \Vol^{t}([-R\sqrt{m-1}, R\sqrt{m-1}]^{t})$.
\end{proof}

\begin{prop}\label{prop:volume.cell+ball.bound}
Let $P=\tconv(A)$ be a polytope in $\mathbb{TP}^{m-1}$ and suppose $A$ has tropical rank $k$.  Assume $m\geq 2$ and $\varepsilon\leq R$.   Then,  $$(\forall d=0,\dots,k-1)(\forall S\in \sS^d) \qquad \Vol^{m-1}( X_S+B^{m-1}_H(\varepsilon))\leq (4R\sqrt{m-1})^{d}(2\varepsilon)^{m-1-d}.$$
\end{prop}
\begin{proof} Fix some bounded cell $X=X_S$ where $S\in\sS^d$ ($\sS^d$ is defined in \eqref{eqn:No.CellsDim}).
Observe that $B^{m-1}_H(\varepsilon)\subseteq [-\varepsilon, \varepsilon]^{m-1}$.  Also, by Proposition \ref{prop:CellsinCube}, $X$ is contained in a translation and/or rotation of $\left[-R\sqrt{m-1}, R\sqrt{m-1}\right]^d$.  
In a simplified notation, we will write $\theta\left[-t, t\right]^d$ to denote some rotation of the $d$-dimensional cube $[-t,t]^d$ in $m-1$ dimensional space.
Thus, for some $\theta$, $X+\HB$ is contained in a translation of $\theta\left[-R\sqrt{m-1}, R\sqrt{m-1}\right]^d+B^{m-1}_H(\varepsilon),$ and therefore \begin{equation}\label{eqn:BoundVolCell1}\Vol^{m-1}(X+\HB)\leq \Vol^{m-1}\left(\theta\left[-R\sqrt{m-1}, R\sqrt{m-1}\right]^d+[-\varepsilon, \varepsilon]^{m-1}\right).\end{equation}

Observe that, to calculate the volume of a Minkowski summand, it doesn't matter which shape the rotation is applied to, i.e.\ for any $t\in\R$ and rotation $\theta$, \begin{equation}\label{eqn:BoundVolCell2}\Vol^{m-1}\left(\theta\left[-t, t\right]^d+[-\varepsilon, \varepsilon]^{m-1}\right)=
\Vol^{m-1}\left(\left[-t, t\right]^d+\theta^{-1}[-\varepsilon, \varepsilon]^{m-1}\right).\end{equation}

Now, since the largest distance between any two points in $[-\varepsilon, \varepsilon]^{m-1}$ is $2\varepsilon\sqrt{(m-1)}$, there exist vectors $u,v$ such that $\theta^{-1}[-\varepsilon, \varepsilon]^{m-1}\subseteq u+ B_2^{m-1}(\varepsilon\sqrt{m-1})\subseteq v+[0, 2\varepsilon\sqrt{m-1}]^{m-1}.$
Therefore \begin{align}\Vol^{m-1}\left(\left[-t, t\right]^d+\theta^{-1}[-\varepsilon, \varepsilon]^{m-1}\right)&\leq\Vol^{m-1}\left(\left[-t, t\right]^d+[0, 2\varepsilon\sqrt{m-1}]^{m-1}\right) \nonumber
\\&=(2t+2\varepsilon\sqrt{m-1})^d(2\varepsilon)^{m-1-d}.\label{eqn:BoundVolCell3}\end{align}

Hence, from \eqref{eqn:BoundVolCell1}, \eqref{eqn:BoundVolCell2},\eqref{eqn:BoundVolCell3} and using our assumption on the size of $\varepsilon$, $$\Vol^{m-1}(X+\HB)\leq (2R\sqrt{m-1}+2\varepsilon\sqrt{m-1})^d(2\varepsilon)^{m-1-d}\leq (4R\sqrt{m-1})^d(2\varepsilon)^{m-1-d}.$$
\end{proof}

\begin{prop}
Let $P=\tconv(A)$ be a polytope in $\mathbb{TP}^{m-1}$ and suppose $A$ has tropical rank $k$.   Assume $m\geq 2$ and $\varepsilon\leq R$.  Then,  $$\Vol^{m-1}(P+B^{m-1}_H(\varepsilon))\leq 2^m3^{m+n-2}(2R\sqrt{m-1})^{k-1}\varepsilon^{m-k}. $$
\end{prop}
\begin{proof}
First, note that the largest cell of $P$ has dimension $k-1$.  Now, using Proposition \ref{prop:volume.cell+ball.bound} and Corollary \ref{prop:bound.cells.P},
\begin{align*}\Vol^{m-1}( P+B^{m-1}_H(\varepsilon))&\leq \sum_{d=0}^{k-1} \sum_{S\in\sS^d} \Vol^{m-1}(X_S+ B^{m-1}_H(\varepsilon))
\\&\leq \sum_{d=0}^{k-1} 3^{m+n-d-2}(4R\sqrt{m-1})^d(2\varepsilon)^{m-1-d}
\\&\leq 3^{m+n-2} \sum_{d=0}^{k-1} (4R\sqrt{m-1})^d(2\varepsilon)^{m-1-d} 
\\&\leq 3^{m+n-2}(4R\sqrt{m-1})^{k-1}(2\varepsilon)^{m-1-(k-1)} \sum_{d=0}^{k-1}\bigg(\frac{2\varepsilon}{4R\sqrt{m-1}}\bigg)^{d}
\\& < 3^{m+n-2}(4R\sqrt{m-1})^{k-1}(2\varepsilon)^{m-k}2 \end{align*}

 since, by our assumptions on the size of $\varepsilon$, $$\sum_{d=0}^{k-1}\bigg(\frac{2\varepsilon}{4R\sqrt{m-1}}\bigg)^{d}\leq\sum_{d=0}^{k-1}\bigg(\frac{1}{2}\bigg)^{d}<2.$$
\end{proof}

\subsection{Lower bound}
Finally we prove the lower bound in Theorem \ref{thm:VolBounds}.  
Let $B_2^d(r)$ be the Euclidean ball of dimension $d$ with radius $r$ (centered at the origin) and $B_2^d$ be the unit Euclidean ball of dimension $d$ (centered at the origin).  We define
\begin{equation}\label{eqn:VolBall}\kappa_d:=\Vol^d(B_2^d)=\frac{\pi^{\frac{d}{2}}}{\Gamma(\frac{d}{2}+1)}.\end{equation}

We will use the following results on the intrinsic volume.  We denote by $V_i(P)$ the $i$th intrinsic volume of $P$, see \cite{ConvPolytope} for the definition. 

\begin{prop}[{\cite{ConvPolytope}}]
\label{prop.intrinsicvolumesum}
Suppose $P\subseteq\mathbb{R}^d$ is a (classical) polytope. Then $$\Vol^d(P+\lambda B_2^d)=\sum_{i=0}^d \lambda^{d-i}\kappa_{d-i} V_i(P).$$
\end{prop}

\begin{rmk}\label{rmk.intrinsicvolume} Intrinsic volumes are unchanged if $P$ is embedded in some Euclidean space of different dimension.  For $\dim{P}=k\leq d$, $V_k(P)$ is the ordinary $k$-volume of $P$ with respect to the Euclidean structure induced in $\aff(P)$ \cite{ConvPolytope}.\end{rmk}

\begin{prop}\label{prop.Hilbertball>Ball}
$B_2^{d-1}(\frac{\varepsilon}{\sqrt{2}})\subseteq B_H^{d-1}(\varepsilon)$.
\end{prop}

\begin{proof}
Recall that $$B_H^{d-1}(\varepsilon):=\{ x\in\mathbb{TP}^{d-1}: \max_{i=2,\dots,d} (x_i-x_1)- \min_{i=2,\dots,d}(x_i-x_1)\leq \varepsilon\}$$ and we assume without loss of generality that $x_1=0$.
We show that, if $x$ is on the boundary of $B_H^{d-1}(\varepsilon)$, then $d_2(x, 0)\geq \frac{\varepsilon}{\sqrt{2}}$ where $d_2$ denotes the standard Euclidean distance and $0$ is the zero vector, this will prove the claim.

Suppose $x\in\TP^{d-1}$ satisfies $\max_{i=2,\dots,d} (x_i)- \min_{i=2,\dots,d}(x_i)= \varepsilon$.  Let $a=\max_{i=2,\dots,d} (x_i)$ and $b=\min_{i=2,\dots,d}(x_i)$.  Then $d_2(x, 0)\geq \sqrt{a^2+b^2}=\sqrt{a^2+(\varepsilon-a)^2}.$
 Now, the equation $a^2+(\varepsilon-a)^2-\frac{\varepsilon^2}{2}=0$ has one root, $a=2\varepsilon$ which implies that $a^2+(\varepsilon-a)^2\geq\frac{\varepsilon^2}{2}$.  Therefore $d_2(x, 0)\geq\sqrt{\frac{\varepsilon^2}{2}}.$ 
\end{proof}

\begin{prop}Let $A\in\Z^{m\times n}$ with $\troprank(A)=k$. Let $P=\tconv(A)\subset\mathbb{TP}^{m-1}$ and $X_T\subseteq\TP^{k-1}$  be defined by \eqref{eqn:XT}Then 

$$ \Vol^{m-1}(P+B^{m-1}_H(\varepsilon))\geq \bigg(\frac{\varepsilon}{\sqrt{2}}\bigg)^{m-k}\kappa_{m-k} k\radius(X_T)^{k-1}.$$
\end{prop}

\proof Since the tropical rank of $P$ is $k$ we have that the dimension of $P$ is $k-1$ and there is a bounded cell, $X_S$, of $P$ of dimension $k-1$. Observe that $$P+\HB\supseteq X_S+\HB\supseteq X_S+\TBall$$ using Proposition \ref{prop.Hilbertball>Ball} and therefore,  \begin{align*}\Vol^{m-1}(P+\HB)&\geq \Vol^{m-1}(X_S+\TBall).\end{align*}
To complete the proof, we bound the latter expression.
Let $X_S\simeq X_T$ where $X_T\subseteq\TP^{k-1}$  is defined by \eqref{eqn:XT}.
Then, using Proposition \ref{prop.intrinsicvolumesum}, Remark \ref{rmk.intrinsicvolume} and Corollary \ref{cor:VolandIntPointsX_S(usingX_T)},
\begin{align*}\Vol^{m-1}(X_S+\TBall)&=\sum_{i=0}^{m-1}\bigg(\frac{\varepsilon}{\sqrt{2}}\bigg)^{m-1-i}\kappa_{m-1-i}V_i(X_S)
\\&\geq \bigg(\frac{\varepsilon}{\sqrt{2}}\bigg)^{m-1-(k-1)}\kappa_{m-1-(k-1)}V_{k-1}(X_S)
\\&=\bigg(\frac{\varepsilon}{\sqrt{2}}\bigg)^{m-k}\kappa_{m-k}\Vol^{k-1}(X_S)
\\&\geq \bigg(\frac{\varepsilon}{\sqrt{2}}\bigg)^{m-k}\kappa_{m-k} k\radius(X_T)^{k-1}.
\end{align*}
\endproof

\section{Hardness results}\label{sec:Hardness}

In this section we prove that approximating the volume of a tropical polytope is hard.  We first show that computing the volume of a tropical polytope is NP-hard by a reduction from the problem of finding the tropical rank of a matrix.  We then extend our proof technique to prove hardness of approximation.

\subsection{Tropical rank}
 
Let us recall the following problem.

\begin{problem} Tropical rank (TROPRANK)

Input:  $A\in\Q^{m\times n}$, $k\in \mathbb{N}$.

Question:  Determine whether $\troprank(P)\geq k$.\end{problem}

Deciding whether a square matrix has maximal tropical rank is polynomially solvable \cite{BH85} but in general the problem of determining the tropical rank is NP-hard for $\{0,\infty\}$ or for $\{0,1\}$ matrices as proved by Kim and Roush \cite{KR05}. 
A detailed explanation of their proof appears in the dissertation \cite{BormanThesis}.  
Shitov gave another proof in \cite{Sh13}. 
This implies that computing the dimension of a tropical
polytope defined by vertices is NP-hard. 
Grigoriev and Podolskii \cite{Grigoriev} proved 
that computing the dimension of a tropical prevariety
is NP-hard, this entails that computing the dimension
of a tropical polytope defined by inequalities
(instead of vertices) is also NP-hard. 

\begin{thm}[{\cite{KR05}, see also~\cite{Sh13}}] Determining whether an $n\times n$ matrix with entries in $\{0,1\}$ has tropical rank at least $k$ is NP-complete.  \end{thm}

\subsection{The decision version of computing the volume}
We define the decision version of the problem of finding the volume of a tropical polytope as follows.

\begin{problem} Tropical Volume (VOLUME(TP))

Input:  $A\in\Q^{m\times n}$, $P=\tconv(A)$, $\alpha\in\Q$.

Question:  Is $\Vol^{m-1}(P)\geq\alpha$?\end{problem}
We will use the following results and definitions.
For two problems $P_1, P_2$, let $P_1\le_p P_2$ mean that there exists a Karp reduction (see \cite{Karp}) from $P_1$ to $P_2$.
Note that, throughout, all logarithms are to the base $2$.

Let $A\in\Z^{m\times n}$ with $\troprank(A)=k$ and $P=\tconv(A)\subset\mathbb{TP}^{m-1}$.  From Theorem \ref{thm:VolBounds}, and using that $\radius(X_T)\geq\frac{1}{k}$ by Proposition \ref{prop:MCMofFullDimensionalCell}, we have that $$\C{-}{k}\varepsilon^{m-k}\leq \Vol^{m-1}(P+B^{m-1}_H(\varepsilon))\leq \C{+}{k}\varepsilon^{m-k}$$ where $W=\max_{i,j}|a_{ij}|$, $\varepsilon\leq R$ and $R\leq 4W$,

\begin{align*}\C{-}{k}&:=\frac{ k\left(\frac{1}{k}\right)^{k-1}\kappa_{m-k}}{2^\frac{m-k}{2}}=\frac{\pi^{\frac{m-k}{2}}}{\Gamma(\frac{m-k}{2}+1)2^\frac{m-k}{2}k^{k-2}}\text{ and}\\ \C{+}{k}&:=2^m3^{m+n-2}(2R\sqrt{m-1})^{k-1}.\end{align*}

Let $\barep$ satisfy \begin{align}\log(\barep)&=\min_k\bigg[ \log(C^-_{m,n,k,W})-\log (C^+_{m,n,k-1,W})\bigg],\text{ equivalently} \label{eqn:barepsilon} \\ \barep&:=\min_{k=1,\dots,m} \frac{C^-_{m,n,k,W}}{C^+_{m,n,k-1,W}}.\nonumber\end{align}

We will assume without loss of generality that $m\geq 2$ and $R\geq 1$.  This is valid since $R\in\mathbb{N}\cup\{0\}$ because $A$ is an integer matrix and, if $R=0$ then, using that $a_{1j}=0$ for all $j$, $$(\forall i\in \{2,\dots,m\}) \max_{j=\{1,\dots,n\}}a_{ij}=\min_{j\in\{1,\dots,n\}}a_{ij}$$ which implies that $\tconv(A)$ is a single vertex.

\subsection{Disjoint intervals}

We next verify a key technical point: for any $\varepsilon<\min(\barep, R)$, the intervals described by the upper and lower bounds on the volume of $P+B_H^{m-1}(\varepsilon)$ for each value of $k$ are disjoint.  That is, the bounds in Theorem \ref{thm:VolBounds} are ordered as shown in Figure \ref{figure:Intervals}. 

\begin{prop}\label{prop:IntervalsOrderedNestled}For all $\varepsilon<\min(\barep, R)$, i.e.\ such that $\log(\varepsilon)<\log(\barep)$, 
\begin{align*}(\forall k\in\{2,\dots,m-1\}) \ &\log(\C{+}{k-1}\varepsilon^{m-(k-1)})<\log(\C{-}{k}\varepsilon^{m-k})\\&<\log(\C{+}{k}\varepsilon^{m-k})<\log(\C{-}{k+1}\varepsilon^{m-(k+1)}).\end{align*}
\end{prop}

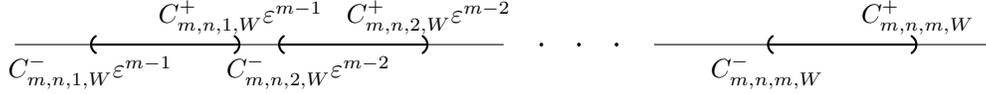
\begin{figure}\begin{center} 
\begin{tikzpicture}
\draw (-1,0)--(5.5,0) (7.5,0)--(12,0);
\draw[(-), thick] (0,0)--(2,0);
\draw (0,0) node[below] {\small $\C{-}{1}\varepsilon^{m-1}$} (2,0) node[above] {\small $\C{+}{1}\varepsilon^{m-1}$};

\draw[(-), thick] (2.5,0)--(4.5,0);
\draw (2.9,0) node[below] {\small $\C{-}{2}\varepsilon^{m-2}$} (4.5,0) node[above] {\small $\C{+}{2}\varepsilon^{m-2}$};

\filldraw (6,0) circle(0.5pt) (6.5,0) circle(0.5pt) (7,0) circle(0.5pt);

\draw[(-), thick] (9,0)--(11,0);
\draw (9,0) node[below] {\small $\C{-}{m}$} (11,0) node[above] {\small $\C{+}{m}$};
\end{tikzpicture}
\end{center}\caption{\small{The intervals representing the possible values of the volume of $P+B_H(\varepsilon)$ for each value of $k$ where $\varepsilon<\barep$.}}\label{figure:Intervals} \end{figure}

\proof First observe that, for any $k$, we have $\log(\C{+}{k}\varepsilon^{m-k})-\log(\C{+}{k-1}\varepsilon^{m-(k-1)})=\log(2R\sqrt{m-1})-\log(\varepsilon)>0.$  Thus \begin{equation}\label{eqn:NestedC+} \dots <\C{+}{k-1}<\C{+}{k}<\C{+}{k+1}<\dots.\end{equation}

By our choice of $\varepsilon$, for all $k$, \begin{equation}\label{eqn:IntervalsDisjoint} \log(\C{+}{k-1}\varepsilon^{m-(k-1)})<\log(\C{-}{k}\varepsilon^{m-k}).\end{equation}

Finally, from Theorem \ref{thm:VolBounds}, \begin{equation}\label{eqn:VolumeBoundsOrder}\log(\C{-}{k}\varepsilon^{m-k})\leq\log(\C{+}{k}\varepsilon^{m-k}). \end{equation}

The result follows from \eqref{eqn:NestedC+}, \eqref{eqn:IntervalsDisjoint} and \eqref{eqn:VolumeBoundsOrder}.
\endproof

\begin{cor}\label{cor:DisjointIntervals}  For $\varepsilon<\min(\barep, R)$ the intervals $[\C{-}{k}\varepsilon^{m-k}, \C{+}{k}\varepsilon^{m-k}]_{k=1,\dots,m}$ are disjoint.
\end{cor}

\subsection{NP-hardness of computing the volume of tropical polytope described by vertices}\label{sec:HardnessVol}

Now we prove NP-hardness of our problems.

\begin{thm}\label{thm:Vol.Intervals.Troprank} Let $A\in\{0,1\}^{m\times n}$ and $P=\tconv(A)\subset\mathbb{TP}^{m-1}$.  Fix $\varepsilon<\min\left(\barep, R\right).$ Then $\troprank(A)=k$ if and only if $$\varepsilon^{m-k}\frac{\kappa_{m-k}}{2^\frac{m-k}{2}k^{k-2}}\leq \Vol^{m-1}(P+B^{m-1}_H(\varepsilon))\leq 2^m3^{m+n-2}(2\sqrt{m-1})^{k-1}\varepsilon^{m-k}.$$
\end{thm}

\proof  Note that $R=1$ for matrices with entries in $\{0,1\}$.
From Theorem \ref{thm:VolBounds} we know that the volume of $P+B^{m-1}_H(\varepsilon)$ will lie in an interval 
$[\C{-}{k}\varepsilon^{m-k}, \C{+}{k}\varepsilon^{m-k}].$  But, from Corollary \ref{cor:DisjointIntervals} these intervals are disjoint for each value of $k$.  The result follows.
\endproof

\begin{thm}VOLUME(TP) is NP-hard.
\end{thm}
\proof
Let $A\in\{0,1\}^{m\times n}$, $k$ be an instance of TROPRANK.
We construct an instance of VOLUME(TP) as follows.
Let $P=\tconv(A)\subset\TP^{m-1}$ and choose $\varepsilon\in\Q$ with $\varepsilon<\min\left(\barep, R\right).$  Note for example we could have 
\begin{align}\varepsilon=&\frac{1}{2^{\frac{m-1}{2}} 2^m3^{m+n-2}(2mR\lceil\sqrt{m-1}\rceil)^{m-2}\lceil\frac{m-1}{2}\rceil!}\label{eqn:ChoiceEpsilon}\\
\leq &\min_{k=1,\dots,m}\frac{\pi^{\frac{m-k}{2}}}{2^{\frac{m-k}{2}} 2^m3^{m+n-2}(2kR\sqrt{m-1})^{k-2}\Gamma(\frac{m-k}{2}+1)}=\barep.\nonumber\end{align} 
Note that ${\varepsilon}$ has a polynomial number of bits since $\log(\varepsilon)=\mathcal{O}(m\log(m)+m+n+m\log(W))$.

Let $M$ be the matrix with columns $A_j+H_k$ for all $j\in\{1,\dots,n\}$, $k\in\{1,\dots,m\}$ where $H$ is given by \eqref{eqn:MatrixHforHilbertBall} and $\tconv(H)=B_H^{m-1}(\varepsilon)$.  By Proposition \ref{prop:HOPB.Generators}, $\tconv(M)=P+B_H^{m-1}(\varepsilon)$ and $M\in\Q^{m\times mn}$.
Further, we can choose some $\alpha\in\Q$ such that $\alpha\in(\C{+}{k-1}\varepsilon^{m-(k-1)}, \C{-}{k}\varepsilon^{m-k})$ in polynomial time.

We consider the volume of the tropical polytope $P+B^{m-1}_H(\varepsilon)$. 
Observe that, by Proposition \ref{prop:IntervalsOrderedNestled}, $$\Vol^{m-1}(P+B^{m-1}_H(\varepsilon))\geq \alpha\Leftrightarrow \Vol^{m-1}(P+B^{m-1}_H(\varepsilon))\geq \C{-}{k}\varepsilon^{m-k}.$$  

Finally, from Theorem \ref{thm:Vol.Intervals.Troprank}, $\troprank(A)\geq k$ if and only if $\Vol^{m-1}(P+B^{m-1}_H(\varepsilon))\geq\alpha$ if and only if $M, \alpha$ is an instance of VOLUME(TP).
Therefore TROPRANK $\le_p$ VOLUME(TP).
\endproof

An immediate consequence of this is the following.
\begin{cor}\label{cor:Volume.Vertices.NPHard} Computing the volume of a tropical polytope described by vertices is NP-hard.
\end{cor}

\subsection{Hardness of approximating the volume of tropical polytopes described by vertices}\label{subsec:APPROXVOL}

The same proof method can be used to show something stronger: that even approximating the volume of a tropical polytope is hard.  

An \emph{approximation algorithm with factor} $\alpha$ for the volume of a tropical polytope is an algorithm which, on input $P=\tconv(A)$, outputs $f(P)$ satisfying $\alpha^{-1}\Vol^{m-1}(P)\leq f(P)\leq \alpha \Vol^{m-1}(P).$

Let $\alpha=2^{\poly(m,n)}$ where $\poly(m,n)$ is some polynomial function in $m$ and $n$.  Throughout this section $f(P)$ is defined to be the output of an approximation algorithm with factor $\alpha$ for the volume of a tropical polytope $P$. 
An immediate consequence of Theorem \ref{thm:VolBounds} is the following.

\begin{thm}\label{thm:ApproxVolBounds} Let $A\in\Z^{m\times n}$ with $\troprank(A)=k$. Let $P=\tconv(A)\subseteq\mathbb{TP}^{m-1}$.   Suppose an approximation algorithm with factor $\alpha$ for the volume of a tropical polytope $Q$ outputs $f(Q)$.

Then $$\alpha^{-1}\C{-}{k}\varepsilon^{m-k}\leq f(P+B^{m-1}_H(\varepsilon))\leq \alpha\C{+}{k}\varepsilon^{m-k}. $$ \end{thm}

Let $\tildeep=\alpha^{-2}\varepsilon$ for some $\varepsilon\in\Q$ satisfying $\varepsilon<\min(\barep, R)$, for example let $\varepsilon$ be as in \eqref{eqn:ChoiceEpsilon}.   Observe that $\tildeep<\alpha^{-2}\barep$ and has a polynomial (in the size of the input) number of bits since $\log(\tildeep)=\mathcal{O}(\poly(m,n)+m\log(m)+m+n+m\log(W))$.    

\begin{prop}\label{prop:ApproxVolIntervals}For all $\varepsilon\leq\tildeep$, the intervals $$\left[\alpha^{-1}\C{-}{k}\varepsilon^{m-k}, \alpha\C{+}{k}\varepsilon^{m-k}  \right]_{k=1,\dots,m}$$ are disjoint.  
\end{prop}
\proof For all $k\in\{2,\dots,m\}$, by our choice of $\varepsilon$, \begin{align*}\frac{\alpha^{-1}\C{-}{k}\varepsilon^{m-k}}{\alpha\C{+}{k-1}\varepsilon^{m-(k-1)}}&\geq \frac{1}{\alpha^{2}\varepsilon}\min_k\frac{\C{-}{k}}{\C{+}{k-1}}
=\frac{\barep}{\alpha^{2}\varepsilon}\geq\frac{\barep}{\alpha^2\tildeep}>1.\end{align*}  The result follows.
\endproof

\begin{prop}\label{prop:ApproxInterval.iff.Troprank}Let $A\in\Z^{m\times n}$ and let $P=\tconv(A)\subseteq\mathbb{TP}^{m-1}$.   Fix $\varepsilon\leq\tildeep$. Then $\troprank(A)=k$ if and only if  $$\alpha^{-1}\C{-}{k}\varepsilon^{m-k}\leq f(P+B^{m-1}_H(\varepsilon))\leq \alpha\C{+}{k}\varepsilon^{m-k}. $$
\end{prop}
\proof Follows from Theorem \ref{thm:ApproxVolBounds} and Proposition \ref{prop:ApproxVolIntervals}.
\endproof

\begin{thm}  There is no polynomial time approximation algorithm of factor $\alpha=2^{\poly(m,n)}$ for the volume of a tropical polytope $P=\tconv(A)$ provided P$\neq$NP. 
\end{thm}
\proof Let $A\in\{0,1\}^{m\times n}$ and $P=\tconv(A)$. 
As before, $M\in\Q^{m\times mn}$ such that $\tconv(M)=P+B^{m-1}_H(\tildeep)$ can be calculated in polynomial time.  

By Proposition \ref{prop:ApproxInterval.iff.Troprank}, any approximation algorithm with factor $\alpha$ for the volume of a tropical polytope applied to $P+B^{m-1}_H(\tildeep)$ would also solve the problem of calculating the tropical rank of $A$.  Hence, under the assumption that P$\neq$NP, no such approximation algorithm can run in polynomial time.
\endproof

\section{Counting integer points in tropical polytopes described by vertices}\label{sec:CountingIntegerPoints}
We aim to show, by a similar argument as for the volume, that counting the number of integer points in a tropical polytope is hard.  This time we construct bounds on the number of integer points on an integer dilation, $sP$, $s\in\N$, of a tropical polytope $P$.

Throughout we have $A\in\Z^{m\times n}$ and $P=\tconv(A)\subseteq\TP^{m-1}$.  Observe that, for any $s\in\N$, $sP=s\tconv(A)=\tconv(sA)$ and it is trivially a tropical polytope.

\subsection{Integer Points in Tropical Polytopes: Lower bound}

We calculate a lower bound on the number of integer points in a tropical polytope. 

\begin{prop}\label{prop:LowerBoundIntPoints}Let $A\in\Z^{m\times n}$, $P=\tconv(A)$, $\troprank(A)=k$.  Let $X_S$ be a cell of dimension $k-1$ in $P$ and $X_T$  be defined by \eqref{eqn:XT}.  Then, for any $s\in\N$,
$$|sP\cap\Z^{m-1}|\geq  (\lfloor s\radius(X_T)\rfloor+1)^{k}-\lfloor s\radius(X_T)\rfloor^k.$$

\end{prop}
\proof Follows from Corollary \ref{cor:VolandIntPointsX_S(usingX_T)} that,\begin{align*}|sP\cap\Z^{m-1}|\geq  |sX_S\cap\Z^{m-1}|\ge(\lfloor\radius(sX_T)\rfloor+1)^{k}-\lfloor\radius(sX_T)\rfloor^k.\end{align*}  Finally, it is trivial that $\radius(sX_T)=s\radius(X_T)$.
\endproof

\subsection{Integer Points in Tropical Polytopes: Upper bound}  

Recall the definition of $R$ \eqref{eqn:DefR}, which is the smallest length of any $(m-1)$-cube containing $P$.  When $P$ is full dimensional, a bound is trivial:
\begin{prop} For any $s\in\mathbb{N}$, $|sP\cap\mathbb{Z}^{m-1}|\leq (sR+1)^{m-1}.$
\end{prop}
\proof
Immediate since $|sP\cap\mathbb{Z}^{m-1}|\leq |[0, sR]^{m-1}\cap\Z^{m-1}|$.
\endproof 

Recall from Proposition \ref{prop:CellsinCube} that, for any bounded cell $X_S$ of dimension $d$, $X_S$ is contained in a translation and/or rotation of $ \left[ -R\sqrt{m-1},  R\sqrt{m-1}\right]^d.$
While this allowed us to bound the volume we need to ensure that the number of integer points is not affected by rotating a $d$ dimensional shape in $m-1$ dimensional space.  This is not a problem for the upper bound since it is clear that, for any $d$-dimensional cube drawn in dimension $m\geq d$ with edge length $l$, the maximum number of integer points will be attained when all the edges are parallel to coordinate axes, and corners are on integer vertices.  Thus, let $K$ be any rotation of a $d$-dimensional cube with edge length $l$, then $$|K\cap\Z^{m}|\leq |[0,l]^d\cap\Z^d|=(l+1)^d.$$

These arguments extend easily to integer dilations, so we conclude the following.

\begin{prop}\label{prop:IntPointUpperBoundCell} Given is a $d$ dimensional ($d\leq m-1$) polytrope $X_S\subseteq\TP^{m-1}$.  Let $C:=\left[ -sR\sqrt{m-1},  sR\sqrt{m-1}\right]^d$.  Then, for all $s\in\N$,
 $$|sX_S\cap\Z^{m-1}|\leq |C\cap\Z^d|= (1+2sR\sqrt{m-1})^d.$$  
\end{prop}

\begin{cor}\label{cor:Upperbound.Intpoints} $A\in\Z^{m\times n}$, $P=\tconv(A)$, $\troprank(A)=k$.  Then, for all $s\in\N$, $$|sP\cap \Z^{m-1}|\leq 3^{m+n-2}\frac{(1+2sR\sqrt{m-1})^{k}-1}{2sR\sqrt{m-1}}.$$
\end{cor}
\proof
Using Proposition \ref{prop:IntPointUpperBoundCell}, Corollary \ref{prop:bound.cells.P}, and the formula for the value of a geometric series,
 \begin{align*}|sP\cap \Z^{m-1}|\leq \sum_{t=0}^{k-1}\sum_{S\in\sS^t}|sX_S\cap\Z^{m-1}|
&\leq \sum_{t=0}^{k-1}3^{m+n-t-2}(1+2sR\sqrt{m-1})^{t}
\\&\leq 3^{m+n-2}\sum_{t=0}^{k-1}(1+2sR\sqrt{m-1})^{t}
\\&= 3^{m+n-2}\frac{1-(1+2sR\sqrt{m-1})^{k}}{1-(1+2sR\sqrt{m-1})}
\\&= 3^{m+n-2}\frac{(1+2sR\sqrt{m-1})^{k}-1}{2sR\sqrt{m-1}}.\end{align*}

\section{Hardness of counting integer points in tropical polytopes described by vertices}\label{sec:HardnessCountingProof}

\subsection{Hardness of counting integer points}
We consider the following problem.
\begin{problem} Counting Integer Points in Tropical Polytopes ($\#$INTPOINTS(TP))

Input:  $A\in\Q^{m\times n}$, $\alpha\in\Q$.

Question:  Is $|\tconv(A)\cap\Z^{m-1}|\geq\alpha$?\end{problem}

Define
\begin{align*}L_{k-1}(s)&:=(\lfloor s\radius(X_T)\rfloor+1)^{k}-(\lfloor s\radius(X_T)\rfloor)^k\text{ and}\\ U_{k-1}(s)&:=3^{m+n-2}\frac{(1+2sR\sqrt{m-1})^{k}-1}{2sR\sqrt{m-1}}.\end{align*}
Then we summarise Proposition \ref{prop:LowerBoundIntPoints} and Corollary \ref{cor:Upperbound.Intpoints} as follows.

\begin{thm}\label{thm:CountingBoundsSummary}Let $A\in\Z^{m\times n}$, $P=\tconv(A)$, $\troprank(A)=k$.  Let $X_S$ be a cell of dimension $k-1$ in $P$ and $X_T$  be defined by \eqref{eqn:XT}.  Then, for any $s\in\mathbb{N}$, $$L_{k-1}(s)\le|sP\cap\Z^{m-1}|\leq U_{k-1}(s)$$
\end{thm}

Now that we have bounds on the number of integer points in polytopes which depend on the tropical rank, we follow the same method as in the proof that volume is NP-hard and prove a reduction from TROPRANK to $\#$INTPOINTS(TP).  We obtain the following results.

\begin{prop}\label{prop:CountIntDisjointInterval} Let $s\in\N$ satisfy $s\equiv 0\mod t, \forall t\in\{1,\dots,k\}$ and $$s\geq \max_{i=1,\dots,k}\frac{3^{m+n-2}(2R\sqrt{m-1})^{i-1}}{\radius(X_T)^{i}}.$$  
Then, for all $t=1,\dots,m-1$, $ L_{t-1}(s)\leq U_{t-1}(s)<L_{t}(s)\leq U_{t}(s).$
In other words, the intervals $[L_{t-1}(s), U_{t-1}(s)]_{t=1,\dots,m-1}$ are disjoint.
\end{prop}

\proof To simplify notation we denote $\radius(X_T)$ by $r$.  
Then, we have that \begin{align*}&\sum_{i=1}^{k} {k\choose i}3^{m+n-2}(2sR\sqrt{m-1})^{i-1}=U_{k-1}(s)\text{ and}
\\& \sum_{i=1}^{k} {k+1\choose i}(\lfloor sr\rfloor)^{i}<\sum_{i=0}^{k} {k+1\choose i}\lfloor \radius(sX_T)\rfloor^{i}=L_k(s).\end{align*}

Since $r=\frac{z}{t}$ for some $t\in\{1\dots,k\}$ by Proposition \ref{prop:MCMofFullDimensionalCell} we have, by our choice of $s$, $sr\in\Z$.  Hence $$3^{m+n-2}(2sR\sqrt{m-1})^{i-1}\leq (sr)^i =\lfloor sr\rfloor^{i}\text{ and } {k\choose i}<{k+1\choose i}\forall i$$ Therefore $U_{k-1}(s)<L_k(s)$.
The result then follows from Theorem \ref{thm:CountingBoundsSummary} and the fact that $L_t(s)$ and $U_t(s)$ are increasing functions.
\endproof

\begin{thm}TROPRANK $\le_p$ $\#$INTPOINTS(TP)
\end{thm}
\proof Let $A\in\{0,1\}^{m\times n}$ be an instance of TROPRANK.
Choose $s\in\N$ big enough to satisfy the constraints in Proposition \ref{prop:CountIntDisjointInterval}.  For example, $s=k!\lceil\tilde{s}\rceil$ for \begin{align*}\tilde{s}=k^k3^{m+n-2}(2R\sqrt{m-1})^{k-1}&\geq\max_{i=1,\dots,k}\frac{3^{m+n-2}(2R\sqrt{m-1})^{k-1}}{\left(\frac{1}{k}\right)^i}\\&\geq \max_{i=1,\dots,k}\frac{3^{m+n-2}(2R\sqrt{m-1})^{i-1}}{\radius(X_T)^{i}},\end{align*} where we use that $\radius(X_T)\geq\frac{1}{k}$ from Proposition \ref{prop:MCMofFullDimensionalCell}.  Observe that $s$ has a polynomial number of bits since $\log(s)=\mathcal{O}(m+n+m\log(m)+m\log(W))$.

Note that $M\in\R^{m\times n}$ with columns $M_j=sA_j$ satisfies $\tconv(M)=s\tconv(A)=sP$.
Taking any $\alpha\in\Q\cap (U_{\troprank(A)-2}(s), L_{\troprank(A)-1}(s))$ we conclude from Theorem \ref{thm:CountingBoundsSummary} and Proposition \ref{prop:CountIntDisjointInterval} that $$|sP\cap\Z^{k-1}|\geq\alpha\Leftrightarrow|sP\cap\Z^{k-1}|\geq L_{k-1}(s)\Leftrightarrow \troprank(A)\geq k.$$  Hence
$$(A, k)\in \text{TROPRANK}\Leftrightarrow (M, \alpha)\in\#\text{INTPOINTS(TP)}.$$
Then, since $\alpha$ can be found in time polynomial in the size of the input, the result follows. 
\endproof

\begin{cor}\label{cor:CountingInt.Vertices.NPHard} Counting the number of integer points in tropical polytopes described by vertices is NP-hard.
\end{cor}

\subsection{Hardness of approximating the number of integer points}
\label{sec:HardnessApproxCounting}

Finally, following the same method as proving that approximating the volume is NP-hard, we summarise here the results which prove that approximating the number of integer points in $P$ is hard.

Let $\alpha=2^{\poly(m,n)}$ where $\poly(m,n)$ is some polynomial function in $m$ and $n$.  Throughout this section $g(P)$ is defined to be the output of an approximation algorithm with factor $\alpha$ for the volume of a tropical polytope $P$. 

An immediate consequence of Theorem \ref{thm:CountingBoundsSummary} is the following.

\begin{thm}\label{thm:ApproxCountingBounds} Let $A\in\Z^{m\times n}$ with $\troprank(A)=k$ and $P=\tconv(A)\subset\TP^{m-1}$.   Suppose an approximation algorithm with factor $\alpha$ for counting the number of integer points of a tropical polytope $Q$ outputs $g(Q)$.

Then $$\alpha^{-1}L_{k-1}(s)\leq g(sP)\leq \alpha U_{k-1}(s). $$ \end{thm}

Let $\bar{s}$ be the smallest natural number such that $$\bar{s}\geq \max_{i=1,\dots,k}\frac{3^{m+n-2}\alpha^2(2R\sqrt{m-1})^{i-1}}{\radius(X_T)^i}.$$

\begin{prop}\label{prop:ApproxCountingIntervals}For all $s\geq \bar{s}$, such that, for all $t\in\{1,\dots,k\}, s\equiv 0\mod t$, the intervals $\left[\alpha^{-1}L_{k-1}(s), \alpha U_{k-1}(s)  \right]_{t=1,\dots,m}$ are disjoint.  
\end{prop}

\proof Let $s\geq\bar{s}$.  For the same reasons as in the proof of Proposition \ref{prop:CountIntDisjointInterval}, and writing $r=\radius(X_T)$,
\begin{align*}\alpha^{-1}L_{k}(s)&>\sum_{r=1}^{k} {k+1\choose i}\alpha^{-1}\lfloor sr\rfloor^{i}=\sum_{r=1}^{k} {k+1\choose i}\alpha^{-1}( sr)^{i}
\\&\geq \sum_{r=1}^{k} \sum_{i=1}^k{k+1\choose i} 3^{m+n-2}\alpha(2 sR\sqrt{m-1})^{i-1}
\\&\geq\sum_{r=1}^{k} \sum_{i=1}^k{k\choose i} 3^{m+n-2}\alpha(2sR\sqrt{m-1})^{i-1}=\alpha U_{k-1}(s).
\end{align*}
The result follows.
\endproof

\begin{prop}\label{prop:ApproxIntervalCounting.iff.Troprank}Let $A\in\Z^{m\times n}$ with $\troprank(A)=k$ and $P=\tconv(A)\subset\TP^{m-1}$.   Fix $s\in\N$ with $s\geq\bar{s}$ and $s\equiv 0\mod t$ for all $t\in\{1,\dots,k\}$. Then $\troprank(A)=k$ if and only if  $$\alpha^{-1}L_{k-1}(s)\leq g(sP)\leq \alpha U_{k-1}(s). $$
\end{prop}
\proof Follows from Theorem \ref{thm:ApproxCountingBounds} and Proposition \ref{prop:ApproxCountingIntervals}.
\endproof

\begin{thm}  There is no polynomial time approximation algorithm of factor $\alpha=2^{\poly(m,n)}$ for counting the number of integer points in a tropical polytope $P=\tconv(A)$ where $A\in \{0,1\}^{m\times n}$ provided P$\neq$NP. 
\end{thm}
\proof Let $A\in\{0,1\}^{m\times n}$ and $P=\tconv(A)$. We can choose some natural number $s\geq \bar{s}$ satisfying $s\equiv 0\mod t$ for all $t\in\{1,\dots,k\}$ which has a polynomial number of bits.  Hence the matrix $M$ such that $\tconv(M)=sP$ can be found in polynomial time.

By Theorem \ref{prop:ApproxIntervalCounting.iff.Troprank}, any approximation algorithm with factor $\alpha$ for counting the number of integer points in a tropical polytope applied to $sP$ would also solve the problem of calculating the tropical rank of $A$.  Hence, under the assumption that P$\neq$NP, no such approximation algorithm can run in polynomial time.
\endproof

\section{Tropical polytopes described by inequalities}\label{sec:SharpPforInequalities}

First defined by Valiant \cite{Valiant}, the complexity class $\#$P is, loosely speaking, the class of the ``counting versions'' of NP problems, i.e., the class of functions that return the number of nondeterministic accepting paths of Turing machines with a polynomial running time. 

Valiant proves that $\#$Monotone-2-SAT is $\#$P-hard \cite{Val79}.  This is the problem of counting all satisfying assignments for any boolean formula $\mathcal{F}=C_1\wedge C_2\wedge\dots\wedge C_r$ where each clause $C_i$ contains exactly $2$ unnegated literals.

Here we consider tropical polytopes described by inequalities: $P=\{x\in\R^n: A\odot x\oplus c\leq B\odot x\oplus d\}.$
\begin{thm}\label{thm:VolumeIneqSharpP} Calculating the volume of a tropical polytope described by inequalities is $\#$P-hard.
\end{thm}

\begin{figure}
\tiny
\begin{tikzpicture}\filldraw[gray!30] (0,-2)--(0,-4)--(-2,-4)--(-2,-2)--(0,-2);
\filldraw[gray!30] (-2,0)--(-4,0)--(-4,-2)--(-2,-2)--(-2,0);
\filldraw[gray!60] (-2,-2)--(-2,-4)--(-4,-4)--(-4,-2)--(-2,-2);
\draw[dotted] (0,-2)--(-4,-2) (-2,0)--(-2,-4);

\filldraw (0,0) circle (2pt) node[anchor=east] {$(0,0)$};
\node at (-3,-5) { $-L\leq x_1\leq -\frac{L}{2}$ };
\node at (-1,-5) { $-\frac{L}{2}\leq x_1\leq 0$ };
\node at (-6,-3) {$-L\leq x_2\leq -\frac{L}{2}$};
\node at (-6,-1) { $-\frac{L}{2}\leq x_2\leq 0$};

\node at (-3,1) { $y_1=$True };
\node at (-1,1) { $y_1=$False };
\node at (2,-3) {$y_2=$True};
\node at (2,-1) { $y_2=$False};

\draw[decorate, decoration={brace, amplitude=10pt,mirror},yshift=-10pt] (-2,-4)--(0,-4);
\draw[decorate, decoration={brace, amplitude=10pt,mirror}, xshift=-10pt] (-4,0)--(-4, -2);
\draw[decorate, decoration={brace, amplitude=10pt,mirror}, yshift=-10pt] (-4,-4)--(-2,-4);
\draw[decorate, decoration={brace, amplitude=10pt}, xshift=-10pt] (-4,-4)--(-4, -2);

\draw[decorate, decoration={brace, amplitude=10pt},yshift=10pt] (-2,0)--(0,0);
\draw[decorate, decoration={brace, amplitude=10pt}, xshift=10pt] (0,0)--(0, -2);
\draw[decorate, decoration={brace, amplitude=10pt}, yshift=10pt] (-4,0)--(-2,0);
\draw[decorate, decoration={brace, amplitude=10pt,mirror}, xshift=10pt] (0,-4)--(0, -2);

\draw[very thick] (-2,-2)--(0,-2)--(0,-4)--(-4,-4)--(-4,0)--(-2,0)--(-2,-2);
\end{tikzpicture}\normalsize\caption{Associating subcubes with boolean vectors}\label{fig:2DSubcubes}\end{figure}
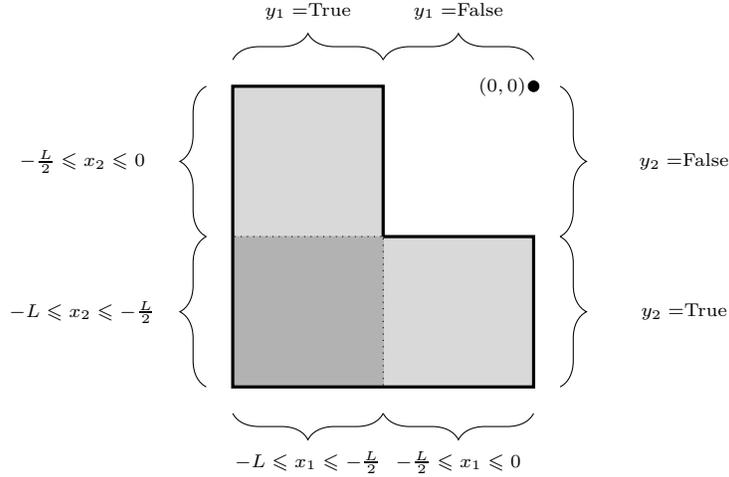

\begin{eg}We first give an example to highlight the idea of the reduction.
Suppose we have the following instance of $\#$Monotone-2-SAT: $F=(x_1\vee x_2)$.  
We define a tropical polytope $P=\{x\in\R^2: -L\leq x_i\leq 0, \frac{-L}{2}\geq \min(x_1, x_2)\}$, and draw $P$ in Figure \ref{fig:2DSubcubes}.
The figure suggests each subcube of $[-L, 0]^2$ can be identified uniquely with a boolean vector.  By construction, the subcube belongs to $P$ if and only if the associated boolean vector is a solution of $\mathcal{F}$.  Hence
the volume of $P$ is $3(\frac{L}{2})^2$ where $3$ is the number of solutions to $\mathcal{F}$.
\end{eg}

\proof The proof is by a reduction from $\#$Monotone-2-SAT.
Let $\mathcal{F}=C_1\wedge C_2\wedge\dots\wedge C_r$ be any instance of Monotone-2-SAT with $m$ clauses and $n$ literals.  Denote the number of satisfying assignments of $\mathcal{F}$ by $\#(\mathcal{F})$ and note that the set of solutions to $\mathcal{F}$ is equivalent to $\{ y\in\{-1,0\}^2: -1\geq\min_{y_i\in C_j} y_i, \forall j=1,\dots,r\}.$

We define, for some $L\in\N$, $$P=\{x\in\R^n: -L\leq x_i\leq 0,  \frac{-L}{2}\geq \min_{x_i\in C_j} x_i \}.$$ It is clear that there exist $A, B, c, d$ such that $P=A\odot x\oplus c\leq B\odot x\oplus d$ and hence $P$ is a tropical polytope.

We claim that $$\Vol^{n}(P)=\#(\mathcal{F})\left(\frac{L}{2}\right)^{n}=\left|\{ y\in\{-1,0\}^2: -1\geq\min_{y_i\in C_j} y_i, \forall j=1,\dots,r\}\right|\left(\frac{L}{2}\right)^{n}.$$  This is since there is a bijection from the solution set of $\mathcal{F}$, to a set of $n$-cubes which form a decomposition of $P$, we explain the detail below.

To each partition $T\cup F$ of $\{1,\dots, n\}$ we identify a subcube of $[-L, 0]^n$ given by
$$\{x\in [-L, 0]^n:  -L\leq x_i\le-\frac{L}{2}, \forall i\in T\text{ and } -\frac{L}{2}\leq x_i\leq 0,\forall i\in F\}.$$

Observe that this partitions $[-L, 0]^n$ into $2^n$ copies of $[-\frac{L}{2},0]^n$.  See also Figure \ref{fig:SubcubeVertices} where subcubes are additionally identified with a unique vertex of $[\frac{L}{2}, 0]^n$, and a partition of $\{1,\dots,n\}$.
The final step is to note that the subcube identified by the partition $T\cup F$ belongs to $P$ if and only if $y$ is a satisfying assignment of $\mathcal{F}$ where $y_j=TRUE\Leftrightarrow j\in T.$  Hence 
$$\bigg|\{y\in\{\text{TRUE, FALSE}\}^n: \mathcal{F}(y)=\text{TRUE}\}\bigg|\left(\frac{L}{2}\right)^n=\Vol^n(P).$$

Therefore computing the volume of a tropical polytope defined by inequalities is at least as hard as any problem in $\#$P.

\begin{figure}
\begin{minipage}{0.5\textwidth}\begin{tikzpicture}[scale=0.4]
\filldraw[gray!30] (0,8)--(-2,6)--(-2,4)--(0,2)--(2,4)--(2,6)--(0,8);
\draw[->] (-6,7)--(-6,8);\draw[->]  (-6,7)--(-5.5, 7.5);\draw[->] (-6,7)--(-5.5, 6.5);
\node[anchor=north west] at (-5.5, 6.5) {$x_1$}; \node[anchor=south west] at (-5.5, 7.5) {$x_2$}; \node[anchor=south] at (-6, 8) {$x_3$};
\foreach \x in{0,2,4}
{
\draw (0, -\x)--(4, -\x+4);
\draw (0, -\x)--(-4, -\x+4);
\draw (\x, \x)--(\x, \x-4);
\draw (-\x, \x)--(-\x, \x-4);
\draw (\x, \x)--(\x-4, \x+4);
\draw (-\x, \x)--(-\x+4, \x+4);
}
\node at (0,8) [anchor=south] {$(-L,0,0)$};
\filldraw (0,8) circle (8pt);
\draw[thick, dashed] (-2,4)--(0,6)--(2,4) (0,6)--(0,8) (2,6)--(2,4)--(0,2)--(0,4) (0,2)--(-2,4)--(-2,6);
\end{tikzpicture}

Partition $T=\{1\}$, $F=\{2,3\}$:

$-L\leq x_1\leq -\frac{L}{2}$, $-\frac{L}{2}\leq x_2, x_3\leq 0$.
\end{minipage}\begin{minipage}{0.5\textwidth}
\begin{tikzpicture}[scale=0.4]

\filldraw[gray!30] (2,6)--(0,4)--(0,2)--(2,0)--(4,2)--(4,4)--(2,6) ;
\foreach \x in{0,2,4}
{
\draw (0, -\x)--(4, -\x+4);
\draw (0, -\x)--(-4, -\x+4);
\draw (\x, \x)--(\x, \x-4);
\draw (-\x, \x)--(-\x, \x-4);
\draw (\x, \x)--(\x-4, \x+4);
\draw (-\x, \x)--(-\x+4, \x+4);
}
\draw[thick, dashed] (2,0)--(0,2)--(0,4);
\node at (4,4) [circle, fill=black] {}; \node at (4,4) [anchor=south] {$(0,0,0)$};
\end{tikzpicture}

Partition $T=\emptyset$, $F=\{1,2,3\}$:

$-\frac{L}{2}\leq x_1, x_2, x_3\leq 0$.\end{minipage}
\caption{Associating subcubes with partitions of $\{1,\dots,n\}$.}\label{fig:SubcubeVertices}\end{figure}
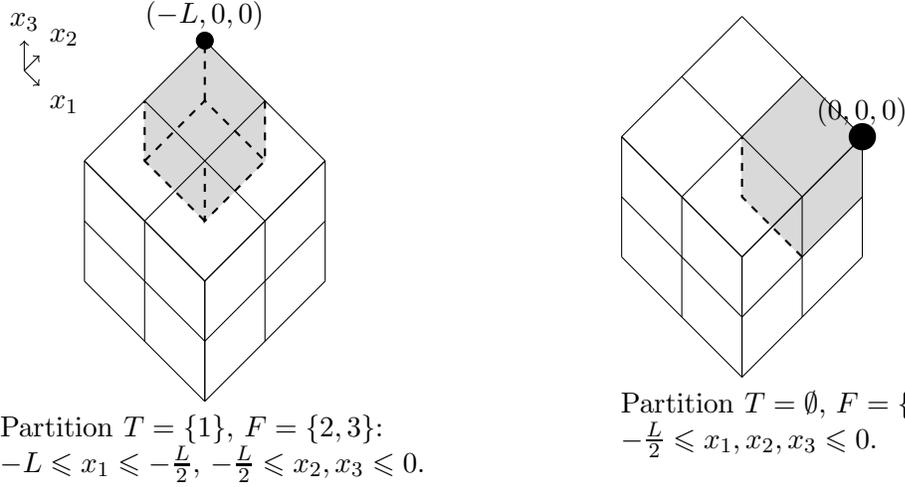
\endproof

\begin{thm} Counting the number of integer points in tropical polytope described by inequalities is $\#$P-hard.
\end{thm}
\proof Similar to the proof of Theorem \ref{thm:VolumeIneqSharpP} we reduce from $\#$Monotone-2-SAT but set $L=1$.  So we have the tropical polytope $$P=\{x\in\R^n: -1\leq x_i\leq 0,  -\frac{1}{2}\geq \min_{x_i\in C_j} x_i \}.$$

Again, we partition the cube $[-1, 0]^n$ into $2^n$ sub-$n$-cubes of length $\frac{1}{2}$.
Then, $y\in\{0,1\}^n$ is a solution to $\mathcal{F}$ if and only if the sub-$n$-cube containing the integer point $x=-y$ belongs to $P$.  Since each sub-$n$-cube contains exactly one integer vertex, and there are exactly $2^n$ integer vertices of $[-1, 0]^n$,  we conclude  
$|\{y\in\{0, 1\}^n: \mathcal{F}(y)=1\}|=|P\cap\Z^n|.$
\endproof

\section{Tropical Polytopes in Fixed Dimension}

We briefly discuss the complexity of the same problems when the dimension is not part of the input.

We begin by discussing the volume.  We have $\Vol^{m-1}(P)=\sum_{s\in\sS^{m-1}}\Vol^{m-1}X_S$ where, by Proposition \ref{prop:number.faces.P}, $$|\sS^{m-1}|\leq {n+m-(m-1)-2 \choose n-(m-1)-1, m-(m-1)-1, m-1}= {n-1 \choose n-m, 0, m-1}=\mathcal{O}\left(\frac{n^m}{m!}\right)$$ which has polynomial size for fixed $m$.  Thus, under our assumption, there are only a polynomial number of cells - which are classical polytopes - whose volume needs to be calculated.  Finally, for fixed dimension, calculating the volume of a classical polytope can be achieved in polynomial time, for detail see, for example, Section 3 of the survey \cite{GK94}.

For counting integer points, it is possible to simply consider each integer point in the translation of $[0,R]^{m-1}$ which contains $P$ (see \eqref{eqn:DefR} and Proposition \ref{prop:CellsinCube} for detail).  For fixed $m$, checking each of $(R+1)^{m-1}$ points is achieved in polynomial time.

We summarise these observations in the following theorem.
\begin{thm} When the dimension is fixed, the following can be computed in polynomial time:
\begin{enumerate}
\item The number of integer points in a tropical polytope with integer generators.
\item The volume of a tropical polytope with integer generators.
\end{enumerate}
\end{thm}

\section{Concluding Remarks}\label{sec-comments}

We showed in Corollaries \ref{cor:Volume.Vertices.NPHard} and \ref{cor:CountingInt.Vertices.NPHard} that, for tropical polytopes described by vertices, VOLUME(TP) and $\#$INTPOINTS(TP) are NP-hard. We suspect the answer to the following question is positive.
\begin{question} Are the following problems $\#$P-hard?
\begin{enumerate}
\item Counting the number of integer points in a tropical polytope described by a list of vertices.

\item Calculating the volume of a tropical polytope described by a list of vertices.
\end{enumerate}
\end{question}

Another open problem concerns the complexity of the same questions when restricting to the class of polytropes.  Polytropes have appeared in several works including~\cite{TropOrdConvexity,DefiniteClosures,KleeneStarPolytropes}.  They are special \emph{alcoved polytopes} corresponding to the root system $A_n$ \cite{AlcovedPolytopes}. Alcoved polytopes are defined by inequalities of the form $\langle a,x\rangle\leq c$ where $a\in\Phi$ for some root system $\Phi$, and $c\in\Z$.  For more detail on alcoved polytopes, see \cite{AlcovedPolytopes}, \cite{AlcovedPolytopes2} and \cite{SymmetricAlcoved}. 
Lam and Postnikov gave in~\cite{AlcovedPolytopes} a combinatorial formula for the normalised volume of an alcoved polytope as the sum of the number of lattice points of other alcoved polytopes, but this does not lead to a polynomial time method. This raises the following problem.
\begin{question}
Determine the complexity of:
\begin{enumerate}
\item Counting the number of integer points of a polytrope.
\item Calculating the volume of a polytrope.
\end{enumerate}
\end{question}

A last open problem, which was one of the initial motivations of the present work, is the question of deciding whether a tropical polytope contains  an integer point.  Determining the existence of an integer point in a classical polytope is NP-hard \cite{GJ79} whereas, for tropical polytopes the complexity is still unknown.  This is called the integer image problem in \cite{BM13}. A pseudopolynomial algorithm is described there, and this problem is also proved to be polynomially solvable in some special cases.  The complexity is further investigated in \cite{Mac16}. 

\section{Acknowledgments}

S. Gaubert was supported by the ANR project MALTHY (ANR-13-INSE-0003), and by ajoint grant of the PGMO program of EDF, Fondation Math\'ematique Jacques Hadamard, and Labex Jacques Hadamard.

M. MacCaig  was supported by a public grant as part of the \emph{Investissement d'avenir} project, reference ANR-11-LABX-0056-LMH, LabEx LHM.

\bibliographystyle{alpha}
\bibliography{tropicalvolume}
\end{document}